\documentclass[11pt]{amsart}
\usepackage{amsthm}

\usepackage[all]{xy}
\usepackage{amssymb}
\usepackage{enumerate}
\usepackage{mathrsfs}
\usepackage{epsfig}
\usepackage{graphicx}
\usepackage{subfig}
\usepackage{float}
\usepackage{epigraph}

\numberwithin{equation}{section}

\newtheorem{thm}{Theorem}[section]

\newtheorem{lem}[thm]{Lemma}

\newtheorem{rem}{Remark}[section]
\newtheorem{example}[thm]{Example}

\newtheorem{defin}[thm]{Definition}

\newcommand{\eq}[1]{(\ref{#1})}

\newcommand{\mbr}{\medbreak}

\renewcommand{\Re}{\operatorname{\rm Re}}
\renewcommand{\Im}{\operatorname{\rm Im}}

\newcommand{\beqast}{\begin{eqnarray*}}
\newcommand{\eqast}{\end{eqnarray*}}
\newcommand{\beqa}{\begin{eqnarray}}
\newcommand{\eqa}{\end{eqnarray}}

\newcommand{\bbe}{\begin{equation}}
\newcommand{\ee}{\end{equation}}

\renewcommand{\Re}{\operatorname{\rm Re}}
\renewcommand{\Im}{\operatorname{\rm Im}}

\newcommand{\bC}{{\mathbb C}}
\newcommand{\bE}{{\mathbb E}}
\newcommand{\bN}{{\mathbb N}}

\newcommand{\bQ}{{\mathbb Q}}

\newcommand{\bR}{{\mathbb R}}

\newcommand{\bZ}{{\mathbb Z}}

\newcommand{\cK}{{\mathcal K}}
\newcommand{\cF}{{\mathcal F}}

\newcommand{\cE}{{\mathcal E}}
\newcommand{\cG}{{\mathcal G}}

\newcommand{\cL}{{\mathcal L}}

\newcommand{\cC}{{\mathcal C}}

\newcommand{\barX}{{\bar X}}
\newcommand{\uX}{{\underline X}}

\newcommand{\cEq}{{\mathcal E_q}}
\newcommand{\cEpq}{{\mathcal E^+_q}}
\newcommand{\cEmq}{{\mathcal E^-_q}}

\newcommand{\phipq}{{\phi^+_q}}
\newcommand{\phimq}{{\phi^-_q}}

\newcommand{\tV}{{\tilde V}}

\newcommand{\hG}{{\hat G}}

\newcommand{\hu}{{\hat u}}
\newcommand{\hU}{{\hat U}}
\newcommand{\hW}{{\hat W}}
\newcommand{\htV}{{\hat{\tilde V}}}

\newcommand{\Om}{{\Omega}}

\newcommand{\be}{\beta}
\newcommand{\De}{\Delta}
\newcommand{\de}{\delta}
\newcommand{\eps}{\epsilon}
\newcommand{\ka}{\kappa}
\newcommand{\la}{\lambda}
\newcommand{\lp}{\lambda_+}
\newcommand{\lm}{\lambda_-}
\newcommand{\La}{\Lambda}
\newcommand{\mum}{\mu_-}
\newcommand{\mup}{\mu_+}
\newcommand{\mumpr}{\mu'_-}
\newcommand{\muppr}{\mu'_+}

\newcommand{\sg}{\sigma}

\newcommand{\om}{\omega}
\newcommand{\omm}{\om_-}

\newcommand{\ze}{\zeta}

\newcommand{\ga}{\gamma}
\newcommand{\gap}{\gamma_+}
\newcommand{\gam}{\gamma_-}
\newcommand{\gappr}{\gamma'_+}
\newcommand{\gampr}{\gamma'_-}

\newcommand{\Ga}{\Gamma}

\newcommand{\bfo}{{\bf 1}}

\begin{document}

\title[Pricing double-barrier options]
{Efficient evaluation of double-barrier options and  joint cpdf of a L\'evy process and its two extrema}
\author[
Svetlana Boyarchenko and
Sergei Levendorski\u{i}]
{
Svetlana Boyarchenko and
Sergei Levendorski\u{i}}

\begin{abstract}
In the paper, we develop a very fast and accurate method for pricing double barrier options with continuous monitoring in
wide classes of  L\'evy models; the calculations are in the dual space, and the Wiener-Hopf factorization
is used. 
 For wide regions in the parameter space, the precision of the order of $10^{-15}$ is achievable in seconds, and of the order of $10^{-9}-10^{-8}$ - in  fractions of a second. The Wiener-Hopf factors and repeated integrals in the pricing formulas
 are calculated using sinh-deformations of the lines of integration, the corresponding changes of variables
 and the simplified trapezoid rule. If the Bromwich integral is calculated using the Gaver-Wynn Rho acceleration
 instead of the sinh-acceleration, the CPU time is typically smaller but the precision is of the order of $10^{-9}-10^{-6}$, at best.
Explicit pricing algorithms and numerical examples are for no-touch options, digitals (equivalently, for
the joint distribution function of a L\'evy process and its supremum and infimum processes), and call options.
Several graphs are produced to explain fundamental difficulties for accurate pricing of barrier options using time discretization
 and interpolation-based calculations in the state space.

\end{abstract}

\thanks{
\emph{S.B.:} Department of Economics, The
University of Texas at Austin, 2225 Speedway Stop C3100, Austin,
TX 78712--0301, {\tt sboyarch@utexas.edu} \\
\emph{S.L.:}
Calico Science Consulting. Austin, TX.
 Email address: {\tt
levendorskii@gmail.com}}

\maketitle



%\tableofcontents

\section{Introduction}\label{s:intro}
Let $X$ be a one-dimensional L\'evy process on the filtered probability space $(\Om, \cF, \{\cF_t\}_{t\ge 0}, \bQ)$
satisfying the usual conditions.
%; in applications to option pricing, the riskless rate is constant and $\bQ$ is an equivalent martingale measure.
  We denote the expectation operator under $\bQ$ by $\bE$. In a number of publications, various methods were applied 
to calculation of expectations $V(f;T;x_1,x_2)$ of functions of spot value $x_1$ of $X$ and its running maximum or minimum $x_2$ and related optimal stopping problems,
standard examples being  barrier and American options, and lookback options with barrier and/or American features. 
See, e.g., \cite{GSh,barrier-RLPE,NG-MBS,BLSIAM02,amer-put-levy,AAP,AKP,KW1,kou,MSdouble,single,BLdouble,beta,KudrLev11,BIL,paraLaplace,HaislipKaishev14,FusaiGermanoMarazzina,kirkbyJCompFinance18,Linetsky08,feng-linetsky09,LiLinetsky2015,BarrStIR,EfficientLevyExtremum,EfficientDiscExtremum} and the bibliographies therein. Options with discrete and continuous monitoring were considered.

In the paper, we develop very fast and accurate method for pricing double barrier options with continuous monitoring in L\'evy models, without rebate: 
if the underlying factor $X$ breaches one of the two barriers
 $h_-<h_+$ before or on the maturity date $T$, the option expires worthless. In neither of the barriers are breached, the payoff
 is $G(X_T)$. The method of the paper can be modified to the case of double barrier options with rebate and to the case of options with discrete monitoring. Pricing barrier options without rebate trivially reduces to the calculation of expectations in a model
  with 0 interest rate $r$, hence, in the main body of the paper, we assume that $r=0$, and allow $\bQ$ to be not an equivalent martingale measure. Explicit algorithms are formulated
  and numerical examples produced for no-touch options, double barrier options with digital payoff (equivalently, joint pdf of $X_T$, 
the infimum process $\uX_t=\inf_{0\le s\le t}X_s$ and supremum
 process $\barX_t=\sup_{0\le s\le t}X_s$), and double barrier call options.

At the first step, we impose no condition on the L\'evy process, and assume that the payoff function $G$ is measurable and bounded. 
Denote by $V(G;h_-,h_+;T,x)$ the price of the option.
 We represent the Laplace transform $\tV(G;h_-,h_+;q,x)$ of the price  in the form
of a sum of the present value of the perpetual stream $G(X_t)$, the discount rate being $q$, and the sum of two series of prices
of perpetual first touch options. 
%We  the terms of the series are evaluated in the iteration procedure. This representation is borrowed from \cite{BLdouble}. 
As in \cite{BLdouble}, we represent the prices
of the perpetual options using the technique of the expected present value operators (EPV operators) developed in a series of publications
\cite{barrier-RLPE,NG-MBS,BLSIAM02,IDUU,single}. The EPV technique is the operator form of the Wiener-Hopf factorization.
In op.cit. as well in a number of other publications where the EPV technique was used, the numerical realization is in the state space. 
 If the calculations are in the state space, the  precision of the order of $10^{-6}$ is very difficult to achieve unless high precision arithmetic is used. In \cite{paraLaplace,paired}, the calculations are in the dual space.
   The fractional-parabolic deformations of the contours of integration at each step of the iteration procedure allow one
   to calculate the integrals in the formulas for the value functions using the simplified trapezoid rule with large but not exceedingly large
   number of terms. In the result, it is possible to
   achieve a precision better that $10^{-6}$. However, in some cases, the accumulated errors of calculation of dozens of thousand of terms
   are too large to achieve a higher precision. 

In the present paper, 
%In the present paper, the calculations are in the dual space, similarly to \cite{paraLaplace,paired},
%where the fractional-parabolic deformations were used. 
we apply a more efficient family of sinh-deformations used in \cite{SINHregular} to
price European options in L\'evy and affine models and in \cite{Contrarian,EfficientLevyExtremum} to price single barrier options and lookback options.  After appropriate conformal deformations of the contours of integration and the corresponding changes of variables, 
we evaluate the resulting repeated integrals applying the simplified trapezoid rule to each integral. The analyticity of the integrands
around the line of integration imply that the error of the infinite trapezoid rule decays as the exponential of $-1/\ze$, where $\ze$ is the step
of the trapezoid rule (see, e.g., Theorem 3.2.1 in \cite{stenger-book}); the sinh-change of variable lead to integrands which decay faster than the exponential function, hence, the truncation error decays faster than the discretization error. 
For the Laplace inversion, we use two methods: the  Gaver-Wynn Rho acceleration algorithm (GWR algorithm), and the sinh-acceleration, as in \cite{EfficientLevyExtremum}; earlier, we applied a less efficient fractional-parabolic deformations \cite{paraLaplace}.
The method of the paper can be regarded as a further step in  a general program of study of the efficiency
of combinations of one-dimensional inverse transforms for high-dimensional inversions systematically pursued by
Abate-Whitt, Abate-Valko \cite{AbWh,AbWh92OR,AbateValko04,AbValko04b,AbWh06} and other authors. Additional methods can be found in \cite{stenger-book}.
The sinh-acceleration is simpler to apply than the saddle-point method (see., e.g., \cite{fedoryuk}), and the calculation of individual terms in
numerical realizations is much simpler and less time consuming. The former method is
more flexible than the latter, in applications to repeated integrals especially. The rates 
of convergence of the sinh-acceleration method and saddle point method are  approximately the same. Talbot's deformation \cite{Talbot79} is not applicable together with the sinh-deformations
of the other contours of integration, hence, the CPU time is significantly larger and good precision is impossible to achieve in many cases where the method of the paper is very efficient.

The rest of the paper is organized as follows. In Section \ref{s:auxil}, we recall the basic facts of the Wiener-Hopf factorization technique,
the definitions of the classes of L\'evy processes amenable to efficient calculations, and efficient formulas for the Wiener-Hopf factors.
In Section \ref{s:double}, we reproduce the iteration procedure for evaluation of the Laplace transform of the price
in terms of the EPV-operators (factors in the operator form of the Wiener-Hopf factorization) derived in \cite{BLdouble}, and then
design efficient procedure in the dual space, which significantly improves the procedure used in \cite{paired}. Explicit algorithms for
the double no-touch, digital and call options are in Section \ref{s:algo_numer}.
Numerical examples are discussed in Section \ref{s:numer}; the figures and tables are relegated to Section \ref{ss:Fig and tables}.
We plot several  graphs to demonstrate fundamental difficulties for accurate pricing of barrier options using time discretization
 and interpolation-based calculations in the state space.
In Section \ref{concl}, we summarize the results and outline possible applications of the method of the paper.

 \section{Auxilliary results}\label{s:auxil}
 
 \subsection{Wiener-Hopf factorization}\label{ss:WHF}
 
 \subsubsection{General formulas for the Wiener-Hopf factors}\label{sss:genWHF}
 For $q>0$, let $T_q\sim\operatorname{Exp}q$ be an
exponentially distributed random variable with mean $q^{-1}$, independent of $X$.
In probability, the 
Wiener-Hopf factors are defined as
\begin{equation}\label{defphipm}
\phi^+_q(\xi)=\bE[e^{i\xi \barX_{T_q}}],\quad \phi^-_q(\xi)=\bE[e^{i\xi \uX_{T_q}}].
\end{equation}
Functions $\phi^\pm_q(\xi)$ appear in the Wiener-Hopf factorization formula
\begin{equation}\label{whfprob}
\frac{q}{q+\psi(\xi)}=\phi^+_q(\xi)\phi^-_q(\xi),\quad \xi\in\bR,
\end{equation}
which is a special case of the Wiener-Hopf factorization of functions of more general classes. 
Define the expected present value operators (EPV-operators) under $X$, $\barX$ and $\uX$ (all three start at 0) by 
$\cE_q u(x)=\bE\left[u(x+X_{T_q})\right]$, $\cE^+_q u(x)=\bE\left[u(x+\barX_{T_q})\right]$ and
$\cE^-_q u(x)=\bE\left[u(x+\uX_{T_q})\right]$. Clearly, the EPV operators are bounded operators in $L_\infty(\bR)$.
In the case of L\'evy processes
with exponentially decaying L\'evy densities, the EPV operators are bounded operators in spaces with exponential weights. For details, see \cite{NG-MBS,single,BLdouble}. 
The operator version of \eqref{whfprob} 
\begin{equation}\label{operWHF}
\cE_q=\cE^+_q\cE^-_q=\cE^-_q\cE^+_q
\end{equation}
is a special case of the operator form of the Wiener-Hopf factorization in the theory of boundary problems for
differential and pseudo-differential operators (pdo). Indeed,  $\cE^\pm_q e^{ix\xi}=\phi^\pm_q(\xi) e^{ix\xi}$.   
 This means that
 $\cE^\pm_q$ are pdo with symbols $\phi^\pm_q$, and $\cE^\pm_qu(x)=\cF^{-1}_{\xi\to x}\phi^\pm_q(\xi)\cF_{x\to\xi}u(x)$ for sufficiently regular functions $u$.  See, e.g., \cite{eskin}.
 The Wiener-Hopf factor $\phi^+_q(\xi)$ (resp., $\phi^-_q(\xi)$) admits analytic continuation to
the half-plane $\{\Im\xi>0\}$ (resp., $\{\Im\xi<0\}$).

The characteristic exponents of all popular classes of L\'evy processes
bar stable L\'evy processes admit analytic continuation to a strip around the real axis. See \cite{NG-MBS,barrier-RLPE,BLSIAM02}, where the general class of Regular L\'evy processes of
exponential type (RLPE) is introduced. 
Let $X$ be a L\'evy process with the characteristic exponent admitting analytic continuation to
a strip $\{\Im\xi\in (\mum,\mup)\}$ around the real axis, and let $q>0$. Then (see., e.g., \cite{NG-MBS,barrier-RLPE,paired})
the following statements hold
\mbr\noindent
I.  There exist
$\sg_-(q)<0<\sg_+(q)$ such that
\begin{equation}\label{crucial}
q+\psi(\eta)\not\in (-\infty,0],\quad \Im\eta\in (\sg_-(q),\sg_+(q)).
\end{equation}
\mbr\noindent
II. The Wiener-Hopf factor $\phi^+_q(\xi)$ admits analytic continuation
to the half-plane $\{\Im\xi>\sg_-(q)\}$, and can be calculated as follows: for any $\om_-\in (\sg_-(q), \Im\xi)$,
\begin{eqnarray}\label{phip1}
\phi^+_q(\xi)&=&\exp\left[\frac{1}{2\pi i}\int_{\Im\eta=\om_-}\frac{\xi \ln (1+\psi(\eta)/q)}{\eta(\xi-\eta)}d\eta\right].
\end{eqnarray}
\mbr\noindent
III. The Wiener-Hopf factor $\phi^-_q(\xi)$ admits analytic continuation
to the half-plane $\{\Im\xi<\sg_+(q)\}$, and can be calculated as follows: for any $\om_+\in (\Im\xi, \sg_+(q))$,
\begin{eqnarray}\label{phim1}
\phi^-_q(\xi)&=&\exp\left[-\frac{1}{2\pi i}\int_{\Im\eta=\om_+}\frac{\xi \ln (1+\psi(\eta)/q)}{\eta(\xi-\eta)}d\eta\right].
\end{eqnarray}
We can (and will) use \eqref{whfprob} and \eq{phip1} to calculate $\phimq(\xi)$ on $\{\Im\xi\in (0,\sg^+_q)\}$, 
and \eqref{whfprob} and \eq{phim1} to calculate $\phipq(\xi)$ on $\{\Im\xi\in (\sg^-_q,0)\}$.

  \subsection{General classes of L\'evy processes amenable to efficient calculations}\label{ss:gen_eff_Levy}
Essentially all popular classes  of L\'evy processes processes enjoy additional properties formalized in \cite{SINHregular,EfficientAmenable}. 

%We borrow following definitions from \cite{EfficientAmenable}.
For $\nu=0+$ (resp., $\nu=1+$), set $|\xi|^\nu=\ln|\xi|$ (resp., $|\xi|^\nu=|\xi|\ln|\xi|$), and introduce the following complete ordering in
the set $\{0+,1+\}\cup (0,2]$: the usual ordering in $(0,2]$; $\forall\ \nu>0, 0+<\nu$; $\forall\ \nu>1, 1<1+<\nu$.
For $\ga\in (0,\pi]$, $\pm\ga_\pm\in (0,\pi2]$ and $\mum<\mup$, define
 $\cC_{\gam,\gap}=\{e^{i\varphi}\rho\ |\ \rho> 0, \varphi\in (\gam,\gap)\cup (\pi-\gap,\pi-\gam)\}$, 
 $\cC_{\ga}=\{e^{i\varphi}\rho\ |\ \rho> 0, \varphi\in (-\ga,\ga)\}$, $S_{(\mum,\mup)}=\{\xi\ |\ \Im\xi\in (\mum,\mup)\}$.

 \begin{defin}\label{def:SINH_reg_proc_1D0}(\cite[Defin 2.1]{EfficientAmenable})
 We say that $X$ is a SINH-regular L\'evy process  (on $\bR$) of   order
 $\nu$ and type $((\mum,\mup);\cC; \cC_+)$
 iff
the following conditions are satisfied:
\begin{enumerate}[(i)]
\item
 $\nu\in\{0+,1+\}\cup (0,2]$ and $\mum<0\le \mup$ or $\mum\le 0<\mup$;
\item
$\cC=\cC_{\gam,\gap}, \cC_+=\cC_{\gampr,\gappr}$, where $\gam<0<\gap$, $\gam\le \gampr\le 0\le \gappr\le \gap$,
and $|\gampr|+\gappr>0$; 
\item
the characteristic exponent $\psi$ of $X$ can be represented in the form
\bbe\label{eq:reprpsi}
\psi(\xi)=-i\mu\xi+\psi^0(\xi),
\ee
where $\mu\in\bR$, and 
$\psi^0$ admits analytic continuation to $i(\mum,\mup)+ (\cC\cup\{0\})$;
\item
for any $\varphi\in (\gam,\gap)$, there exists $c_\infty(\varphi)\in \bC\setminus (-\infty,0]$ s.t.
\begin{equation}\label{asympsisRLPE}
\psi^0(\rho e^{i\varphi})\sim  c_\infty(\varphi)\rho^\nu, \quad \rho\to+\infty;
\end{equation}
\item
the function $(\gam,\gap)\ni \varphi\mapsto c_\infty(\varphi)\in \bC$ is continuous;
% function, which extends to the continuous on the closure of $\cC\cap \{\xi\ |\ |\xi|=1, \Re\xi>0\}$;
\item
for any $\varphi\in (\gampr, \gappr)$, $\Re c_\infty(\varphi)>0$.
\end{enumerate}
\end{defin}
To simplify the constructions in the paper, we assume that $\mum<0<\mup$ and $\gampr<0<\gappr$. 
\begin{example}\label{ex:KoBoL}{\rm  In \cite{genBS,KoBoL}, we constructed a family of pure jump processes
generalizing the class of  \cite{koponen},  with the L\'evy measure
\begin{equation}\label{KBLmeqdifnu}
F(dx)=c_+e^{\lm x}x^{-\nu_+-1}\bfo_{(0,+\infty)}(x)dx+
 c_-e^{\lp x}|x|^{-\nu_--1}\bfo_{(-\infty,0)}(x)dx,
\end{equation}
where $c_\pm>0, \nu_\pm\in [0,2), \lm<0<\lp$.  If $\nu_\pm\in (0,2), \nu_\pm\neq 1$,
\bbe\label{KBLnupnumneq01}
\psi^0(\xi)=c_+\Ga(-\nu_+)((-\lm)^{\nu_+}-(-\lm-i\xi)^{\nu_+})+c_-\Ga(-\nu_-)(\lp^{\nu_-}-(\lp+i\xi)^{\nu_-}).
\ee
 A specialization
 $\nu_\pm=\nu\neq 1$, $c=c_\pm>0$, of KoBoL used in a series of numerical examples in \cite{genBS} was named CGMY model in \cite{CGMY} (and the labels were changed:
 letters $C,G,M,Y$ replace the parameters $c,\nu,\lm,\lp$ of KoBoL):
 \bbe\label{KBLnuneq01}
 \psi^0(\xi)= c\Ga(-\nu)[(-\lm)^{\nu}-(-\lm- i\xi)^\nu+\lp^\nu-(\lp+ i\xi)^\nu].
\ee
Evidently, $\psi^0$ given by \eq{KBLnuneq01} is analytic in $\bC\setminus i\bR$, and  $\forall\ \varphi\in (-\pi/2,\pi/2)$, \eq{asympsisRLPE} holds with
\bbe\label{ascofnupeqnumcc}
c_\infty(\varphi)=-2c\Ga(-\nu)\cos(\nu\pi/2)e^{i\nu\varphi}.
\ee
}
\end{example} 
In \cite{EfficientAmenable}, we defined a class of Stieltjes-L\'evy processes (SL-processes).  Essentially, $X$ is called a (signed) SL-process if $\psi$ is of the form
\bbe\label{eq:sSLrepr}
\psi(\xi)=(a^+_2\xi^2-ia^+_1\xi)ST(\cG^0_+)(-i\xi)+(a^-_2\xi^2+ia^-_1\xi)ST(\cG^0_-)(i\xi)+(\sg^2/2)\xi^2-i\mu\xi, 
\ee
where $ST(\cG)$ is the Stieltjes transform of the (signed) Stieltjes measure $\cG$,  $a^\pm_j\ge 0$, and $\sg^2\ge0$, $\mu\in\bR$.
A (signed) SL-process is called regular if it is SINH-regular. We proved in \cite{EfficientAmenable} that 
the characteristic exponent $\psi$ of 
a (signed) SL-process admits analytic continuation to the complex plane with two cuts along the imaginary axis.
If $X$ is SL-process, then, for any $q>0$, equation $q+\psi(\xi)=0$ has no solution on $\bC\setminus i\bR$.
We proved that all popular classes of L\'evy processes bar the Merton model and Meixner processes are regular SL-processes, with $\ga_\pm=\pm \pi/2$;
the Merton model and Meixner processes are regular signed SL-processes, and $\ga_\pm=\pm \pi/4$.
In \cite{EfficientAmenable}, the reader can find a list of SINH-processes and SL-processes, with calculations of the order and type.

\subsection{Evaluation of  the Wiener-Hopf factors and sinh-acceleration}\label{ss:expl_WHF}%: the case of $q>0$
The integrands on the RHSs of \eq{phip1} and \eq{phim1} decay very slowly at infinity, hence,
it is impossible to achieve a good precision without additional tricks. If $X$ is SINH-regular, the rate of decay of the
integrands can be significantly increased using appropriate conformal deformations of the line of integration
 and the corresponding changes of variables. Assuming that in Definition \ref{def:SINH_reg_proc_1D0},  $\ga_\pm$ are not extremely small is absolute value (and, in the case of  regular SL-processes, $\ga_\pm=\pm \pi/2$ are not small), the most efficient change of variables
is the sinh-acceleration 
   \bbe\label{eq:sinh}
 \eta=\chi_{\om_1,b,\om}(y)=i\om_1+b\sinh(i\om+y), 
\ee
where $\om\in (-\pi/2,\pi/2)$, $\om_1\in \bR, b>0$.
Typically, the sinh-acceleration is the best choice even if $|\ga_\pm|$ are of the order of $10^{-5}$.  The parameters $\om_1,b,\om$ are chosen so that the contour $\cL_{\om_1,b,\om}:=\chi_{\om_1,b,\om}(\bR)\subset i(\mup,\mup)+\cC_{\gam,\gap}$ and, in the process of deformation, $\ln(1+\psi(\eta)/q)$ is a well-defined analytic function on a domain in $\bC$ or an appropriate Riemann surface.  \begin{lem}\label{lem:WHF-SINH}
Let $X$ be SINH-regular of type $((\mum,\mup), \cC_{\gam,\gap}, \cC_{\gampr,\gappr})$. 

Then  there exists $\sg>0$ s.t. for all $q>\sg$, 
\begin{enumerate}[(i)]
\item
$\phipq(\xi)$ admits analytic continuation to $i(\mum,+\infty)+i(\cC_{\pi/2-\gam}\cup\{0\})$. For any $\xi\in i(\mum,+\infty)+i(\cC_{\pi/2-\gam}\cup\{0\})$, and any contour 
$\cL^-_{\om^-_1,b^-,\om^-}\subset i(\mum,\mup)+(\cC_{\gam,\gap}\cup\{0\})$ lying below $\xi$,
\bbe\label{phipq_def}
\phipq(\xi)=\exp\left[\frac{1}{2\pi i}\int_{\cL^-_{\om^-_1,b^-,\om^-}}\frac{\xi \ln (1+\psi(\eta)/q)}{\eta(\xi-\eta)}d\eta\right];
\ee
\item
$\phimq(\xi)$ admits analytic continuation to $i(-\infty,\mup)-i(\cC_{\pi/2+\gap}\cup\{0\})$. For any $\xi\in i(\-\infty,\mup)-i(\cC_{\pi/2+\gap}\cup\{0\})$, and any contour 
$\cL^+_{\om^+_1,b^+,\om^+}\subset i(\mum,\mup)+(\cC_{\gam,\gap}\cup\{0\})$ lying above $\xi$,
\bbe\label{phimq_def}
\phimq(\xi)=\exp\left[-\frac{1}{2\pi i}\int_{\cL^+_{\om^+_1,b^+,\om^+}}\frac{\xi \ln (1+\psi(\eta)/q)}{\eta(\xi-\eta)}d\eta\right].
\ee
\end{enumerate}
\end{lem}
\begin{rem}\label{rem:SL-WHF}{\rm 
\begin{enumerate}[(a)]
\item
In the process of deformation, the expression $1+\psi(\eta)/q$ may not assume value zero. In order to avoid complications stemming from analytic continuation to an appropriate Riemann surface, it is advisable to ensure that $1+\psi(\eta)/q\not\in(-\infty,0]$.
 Thus, if $X$ is a SL-process and $q>0$ - and only positive $q$'s are used in the Gaver-Stehfest method or GWR algorithm -  
any $\om^\pm\in (-\pi/2,\pi/2)$ is admissible in \eq{phipq_def} and \eq{phimq_def} provided $\om^\pm_1$ and $b^\pm$ are chosen so that
the contours are located as stated in Lemma \ref{lem:WHF-SINH}.
\item
For evaluation of certain expectations (prices), it is possible to use contours $\cL^\pm:=\cL^\pm_{\om^\pm_1,b^\pm,\om^\pm}$
with $\om^\pm$ of the same sign. For the iteration procedure in the present paper, it is crucial to use $\cL^+$ and $\cL^-$
in the upper and lower half-planes, respectively.  
If the sinh-acceleration is applied to the Bromwich integral as well,
then additional conditions on the process
and the parameters of the deformations must be imposed. See Lemma \ref{lem:cones_Brom} below.
\item
If $\cL^+_{\om^+_1,b^+,\om^+}$ and $\cL^-_{\om^-_1,b^-,\om^-}$ are in the upper and lower half-planes, respectively, 
and $\phi^\mp_q(\xi)$ for $\xi\in \cL^\pm_{\om^\pm_1,b^\pm,\om^\pm}$ are needed, then 
it is advantageous to calculate $\phi^\pm_q(\xi)$, and then
\beqa\label{eq:rec_phimq}
\phimq(\xi)&=&\frac{q}{\phipq(\xi)(q+\psi(\xi))},\ \xi\in \cL^+_{\om^+_1,b^+,\om^+},\\\label{eq:rec_phipq}
\phipq(\xi)&=&\frac{q}{\phimq(\xi)(q+\psi(\xi))},\ \xi\in \cL^-_{\om^-_1,b^-,\om^-}.
\eqa

\end{enumerate}
}
\end{rem}
%\subsection{Evaluation of  the Wiener-Hopf factors and sinh-acceleration: the case of complex $q$}\label{ss:expl_WHF_comp}
The following lemma is a Lemma 4.1 in \cite{EfficientLevyExtremum} restricted to SINH-regular processes.
\begin{lem}\label{lem:cones_Brom}
Let 
$X$ be a SINH-regular process of order $\nu\in (0+2]\cup\{0+\}$ and type
 \\    $([\mum,\mup], \cC_{\gam,\gap},\cC_{\gampr,\gappr})$, where $\mum<0<\mup$ and $\gampr<0<\gappr$.
   Let either
     $\nu\ge 1$ or $\nu<1$ and the drift $\mu$ in \eq{eq:reprpsi} is 0.
     
     Then $\exists$ $\om_\ell\in (0,\pi/2)$, $c, \sg>0$ such that $\forall$\ 
     $q\in \sg+\cC_{\pi/2+\om_\ell}$ and  $\xi\in i(\mum,\mup)+(\cC_{\gampr,\gappr}\cup\{0\})$,
     \bbe\label{bound:two cones}
     |q+\psi(\xi)|\ge c(|q|+|\xi|^\nu).
     \ee
     \end{lem}
     As in \cite{EfficientLevyExtremum},
       the bound \eq{bound:two cones}
 allows us to use
     one sinh-deformed contour in the lower half-plane and  one in the upper half-plane for all purposes: the calculation 
     of the Wiener-Hopf factors and evaluation of the integrals in the pricing formulas. If either $\mum=0$ or $\mup=0$, then both contours
     must cross $i\bR$ in the same half-plane but the types of contours (two non-intersecting contours, one with the wings deformed upwards, the other one with the wings deformed downwards) remain the same as in the case $\mum<0<\mup$. To simplify 
     the exposition of the main idea, we assume that $\mum<0<\mup$
     
We deform the contour $\{\Re q=\sg\}$ in the Bromwich integral
into a contour of the form $\cL^L=\chi_{L; \sg_\ell, b_\ell, \om_\ell}(\bR)$, where the conformal mapping
$\chi_{L \sg_\ell, b_\ell, \om_\ell}$ is defined by
 \bbe\label{eq:sinhLapl}
\chi_{L; \sg_\ell,b_\ell,\om_\ell}(y)=\sg_\ell +i b_\ell\sinh(i\om_\ell+y),
\ee
and $\sg_\ell, b_\ell>0$, $\om_\ell\in (0,\pi/2)$,  $\sg_\ell-b_\ell\sin\om_\ell>0$. For $q\in \cL^L$ (and $q$'s arising in the process
of deformation of the line $\{\Re q=\sg\}$ into $\cL^L$), we can calculate $\phipq(\xi), \xi\in \cL^-,$ using  \eq{phipq_def},
and $\phimq(\xi), \xi\in \cL^+,$ using  \eq{phimq_def}.

\begin{rem}\label{rem:recip_phipm}{\rm
If \eq{bound:two cones} holds, then, for $q\in \cL^L$,  we can use \eq{eq:rec_phimq}-\eq{eq:rec_phipq}.
}
\end{rem}

\begin{rem}\label{rem:nuless1GWR}{\rm
If $\nu<1$ and $\mu\neq 0$, then the sinh-acceleration cannot be applied to the Bromwich integral if the deformations
$\cL^\pm$ of the contours of integration is used. See \cite{EfficientLevyExtremum}. In this case, only the GWR acceleration
or an acceleration of Euler type  can be used.
}
\end{rem}

  \section{Double barrier options and joint cpdf of $X_T, \barX_T$ and $\uX_T$}\label{s:double}
  \subsection{General scheme in the state space}\label{ss:gen_scheme_state} Let $X$ be a L\'evy process on $\bR$.
  For $h\in \bR$, denote by $\tau^+_h$ and $\tau^-_h$ the first entrance time of $X$ into $[h,+\infty)$ and
 $(-\infty, h]$, respectively.    Let $T>0$ and $h_-<x<h_+$, and let $G\in L_\infty((h_-,h_+))$. 
   Consider 
 \bbe\label{Vdoublent1}
 V(G;h_-,h_+;T,x)=\bE^x[\bfo_{\tau^-_{h_+}\wedge \tau^+_{h_+}>T}G(X_T)].
 \ee 
 Let $q>0$ and let $T_q$ be an exponentially distributed 
 random variable of mean $1/q$, independent of $X$. Assume that $X_0=\barX_0=\uX_0=0$.
 Then,
 for $x\in (h_-,h_+)$, the Laplace transform $\tV(G; h_-,h_+;q,x)$ of $V(G; h_-,h_+;T,x)$ can be represented as 
 \bbe\label{eq:repr1}
 \tV(G;h_-,h_+;q,x)=q^{-1}\bE^x[\bfo_{\tau^-_{h_+}\wedge \tau^+_{h_+}>T_q}G(X_{T_q})].
 \ee
 Let $lG$ be a bounded measurable extension of $G$ to $\bR$. 
Set $\tV^0(lG; q,\cdot)=q^{-1}\cEq lG$, and consider  $\tV^1(lG;h_-,h_+;q,x):=\tV(G;h_-,h_+;q,x)-\tV^0(lG; q,x)$.
We calculate $\tV^1(lG;h_-,h_+;q,x)$ in the form of a series exponentially converging in $L_\infty$-norm. 
The terms of the series depend on the choice of the extension but $\tV(G;h_-,h_+;q,x)$ is independent of
the choice.  
%Furthermore, in exponential L\'evy 
%This trivial observation is useful to simplify the numerical algorithm for such options as the double barrier call option:
%the natural extension of the payoff function from $[h_-,h_+]$ is unbounded. 
Define
 \beqa\label{def:tVp1}
\tV^+_1(lG;h_-,h_+;q,x)=\bE^x[e^{-q\tau^+_{h_+}}\tV^0(lG;q, X_{\tau^+_{h_+}})],\
\\
\label{def:tVm1}
\tV^-_1(lG;h_-,h_+;q,x)=\bE^x[e^{-q\tau^-_{h_-}}\tV^0(lG;q, X_{\tau^-_{h_-}})],
\eqa
and note that  $\tV^+_1$ (resp., $\tV^-_1$) is the  EPV of the stream $G(X_t)$ which starts to accrue the first time $X_t$ crosses $h_+$ from below
(crosses $h_-$ from above). Inductively, for $j=2,3,\ldots,$ define
\beqa\label{def:tVpj}
\tV^+_j(lG;h_-,h_+;q,x)=\bE^x[e^{-q\tau^+_{h_+}}\tV^-_{j-1}(lG;h_-,h_+;q, X_{\tau^+_{h_+}})],\
\\
\label{def:tVmj}
\tV^-_j(lG;h_-,h_+;q;x)=\bE^x[e^{-q\tau^-_{h_-}}\tV^+_{j-1}(lG;h_-,h_+;q, X_{\tau^-_{h_-}})].
\eqa
In \cite{BLdouble}, the following key theorem is derived from the stochastic continuity of $X$:
\begin{thm}\label{thm:convergence}
%\begin{enumerate}[(a)]
%\item
For any $\sg>0$ and $h_-<h_+$, there exist $\de_\pm=\de_\pm(\sg, h_+-h_-)\in (0,1)$ such that for all 
$q\ge \sg$, $lG\in L_\infty(\bR)$ and $j=2,3,\ldots$, 
\beqa\label{eq:boundtVp}
\sup_{x\ge h_+}|\tV^+_j(lG;h_-,h_+;q;x)|&\le&\de_+\sup_{x\le h_-}|\tV^-_{j-1}(lG;h_-,h_+;q,x)|,\\
\label{eq:boundtVm}
\sup_{x\le h_-}|\tV^-_j(lG;h_-,h_+;q;x)|&\le&\de_- \sup_{x\ge h_+}|\tV^+_{j-1}(lG;h_-,h_+;q,x)|,
\eqa
%\item 
and
\beqa\label{qtV}
 \tV^1(lG;h_-,h_+;q,x)&=&\sum_{j=1}^{+\infty} (-1)^j (\tV^+_j(lG;h_-,h_+;q,x)+\tV^-_j(lG;h_-,h_+;q,x)).
 \eqa
 The series on the RHS of \eq{qtV} exponentially converges in $L_\infty$-norm.

% \item
 
 %\end{enumerate}
\end{thm}
%On the strength of \eq{eq:boundtVp}-\eq{eq:boundtVm}, 
Under additional weak conditions on $X$, Theorems 11.1.4 and 11.1.5 in \cite{IDUU} state that 
\beqa\label{tVhm1}
\tV^-_1(lG; h_-,h_+; q,x)&=&q^{-1}(\cEmq\bfo_{(-\infty,h_-]}\cEpq)lG)(x),\\
\label{tVhp1}
\tV^+_1(lG; h_-,h_+; q,x)&=&q^{-1}(\cEpq\bfo_{[h_+,+\infty)]}\cEmq)lG)(x);
\eqa
in \cite{single}, \eq{tVhm1}-\eq{tVhp1} are proved for any L\'evy process.
Similar representations for $\tV^\mp_j$, $j=2,3,\ldots$:
\beqa\label{tVhpj2}
\tV^+_j(lG; h_-,h_+; q,x)&=&\cE^+_q\bfo_{[h_+,+\infty)}(x)(\cE^+_q)^{-1}\tV^-_{j-1}(lG;h_-,h_+; q,x),
\\\label{tVhmj2}
\tV^-_j(lG;h_-,h_+; q,x)&=&(\cE^-_q\bfo_{(-\infty, h_-]}(\cE^-_q)^{-1}\tV^+_{j-1})(lG;h_-,h_+; q,x),
\eqa
follow from Theorems 11.1.6 and 11.1.7 in \cite{IDUU}
provided we can prove that $\tV^\pm_{j-1}$, $j=2,3,\ldots,$ admit the representations
\beqa\label{rprtVpj}
\tV^+_{j-1}(lG; h_-,h_+; q,x)&=&q^{-1}(\cEq G^+_{j-1})(x), \ x\ge h_+,\\
\label{rprtVmj}
\tV^-_{j-1}(lG; h_-,h_+; q,x)&=&q^{-1}(\cEq G^-_{j-1})(x), \ x\ge h_-,
\eqa
where $G^\mp_{j-1}\in L_\infty(\bR)$.
We consider L\'evy processes with the characteristic exponents of class $C^\infty$ and regular asymptotic behavior at infinity. In this case, it follows from the general regularity results for solutions of boundary problems for pdo (see \cite{eskin}
for a general theory, 
and \cite{NG-MBS} for a simpler direct proof in the one-dimensional case) that
$(q-L_X)V^+_{j-1}=q(\cEq)^{-1}V^+_{j-1}$,
is infinitely differentiable on $(h_-,+\infty)$ and all derivatives exponentially decay at infinity. Hence, the representation
\eq{rprtVpj} exists. By symmetry, the representation
\eq{rprtVmj} exists as well.

\begin{rem}\label{rem:convergence}{\rm 
\begin{enumerate}[(a)]
\item
  For a numerical realization, the series on the RHS of \eq{qtV} are truncated,
and $\sum_{j=1}^\infty$ replaced with $\sum_{j=1}^{M_0}$; given the error tolerance,
 $M_0$ can be chosen using the bounds in Theorem \ref{thm:convergence}.

\item If $G$ is given by an analytical expression that defines an unbounded function on $\bR$ but bounded on $[h_-,h_+]$, we can use the scheme above replacing $G$ with a function $lG$ of class $L_\infty(\bR)$, which coincides with $G$ on $[h_-,h_+]$. This simple consideration suffices for a numerical realization  in the state space.
If the numerical realization is in the dual space, complex-analytical properties of the Fourier transform $\hG$ are crucial.
If $\hG$ has good properties as in the case of the call option, then a good replacement  of $G$   with a bounded function
requires additional work. Instead, we will not replace $G$ but assume that the strip of analyticity of $\psi$ is sufficiently wide, and
$q>0$ is sufficiently large so that the series \eq{qtV} converges in a space with an appropriate exponential weight. 

\item Alternatively, in the case of a call option, one can make an appropriate Esscher transform and reduce to the case of 
a bounded $G$.

\item
The inverse Laplace transform of $\tV^0(q,x)=q^{-1}\cEq G(x)$ is the price $V_{\mathrm{euro}}(G; T,x)$ of the European option
 with the payoff $G(x+X_T)$ at maturity $T$. An explicit procedure for an efficient numerical evaluation of $V_{\mathrm{euro}}(G; T,x)$
 can be found in \cite{SINHregular}. In the paper, we design an efficient numerical procedure for the evaluation of 
 $V^1(G; h_-,h_+;T,x)= V(G; h_-,h_+;T,x)-V_{\mathrm{euro}}(G; T,x)$. Note that  $V^1(G; h_-,h_+;T,x)$ is the inverse Laplace transform of the series
 on the RHS of \eq{qtV}.
 \item
 To shorten the notation, below, we suppress the dependence of $\tV^0, \tV^1$ and $\tV^\pm_j$ on $lG$.
 \end{enumerate}
}
\end{rem}
\subsection{Calculations in the dual space}\label{ss:calc_dual}
We make the calculations in the dual space, with the exception of the last step, when the inverse Fourier transform is applied to evaluate the final result. Also, we assume that $X$ is SL-regular.   Then we can choose the contours $\cL^\pm$  
and evaluate the Wiener-Hopf factors $\phi^\pm_q(\xi)$ exactly as in \cite{EfficientLevyExtremum}. In the case $q>0$,
the contours are of the form $\cL^\pm=\chi_{\om^\pm_1,b^\pm,\om^\pm}(\bR)$, where $\pm\om^\pm\in (0,\pi/2)$, and $b^\pm>0$ and $\om^\pm_1$ are
chosen so that $\pm(\om^\pm_1+b^\pm\sin (\om^\pm))>0$, and $i(\om^\pm_1+b^\pm\sin (\om^\pm))$ are on the open interval  on the imaginary axis around zero, where $q+\psi(\xi)>0$.
If the sinh-acceleration is applied to  the Bromwich integral, then the choice of the parameters of the deformation is more involved.
In all cases, we use the same deformations and contours $\cL^\pm$ as in \cite{EfficientLevyExtremum}.
The detailed recommendations are the same. In the numerical scheme below, we calculate the Fourier transforms $\htV^+_1(h_-,h_+;q,\xi)$
on $\cL^-$, and the Fourier transforms $\htV^-_1(h_-,h_+;q,\xi)$
on $\cL^+$, and then, by induction, all $\htV^\pm_j(h_-,h_+;q,\xi)$
on $\cL^\mp$. Only the first step is payoff-dependent.

\subsubsection{The first step}\label{sss:first_step}
 Analytic continuation of $\tV^0(q;\xi)$ and $\htV^\pm_1(h_-,h_+;q,\xi)$ to sufficiently large domains is possible under additional conditions
on $G$. We consider three basic cases.
\begin{enumerate}[(a)]
\item
{\em Double no-touch option} $V_{n.t.}(h_-,h_+;T,x)=V(1;h_-,h_+;T,x)$. We have  $\tV^0(q,x)=1/q$, 
\bbe\label{htV1p}
\htV^\pm_1(h_-,h_+;q,\xi)=\pm e^{-ih_\pm\xi}\frac{\phi^\pm_q(\xi)}{i\xi q}, \ \xi\in \cL^\mp.
\ee
Define $\hW^\pm_1(h_-,h_+,q,\xi)=\pm 1/(i\xi)=\mp i/\xi$, $\xi\in \cL^\mp$.
\item
{\em Digital no-touch option, equivalently, the joint cpdf of} $X_T, \barX_T, \uX_T$. For $a\in [h_+,h_-]$,
$G(x)=G(a;x)=\bfo_{(-\infty,a]}(x)$. Under our assumptions on $X$, $\tV(G(h_-;\cdot);h_-,h_+;q,x)=0$, and 
$\tV(G(h_+;\cdot);h_-,h_+;q,x)=V_{n.t.}(h_-,h_+;T,x)$, $x\in (h_-,h_+)$. Hence, it suffices to consider $a\in (h_-,h_+)$.
%Making the change of variables if necessary, we may assume that $a=0$. 
 Using the operator form of the Wiener-Hopf factorization, we simplify \eq{tVhp1} and \eq{tVhm1}:
 \beqa\label{jointtV1p2}
\tV^+_1(h_-,h_+;q,x) &=&q^{-1}(\cEpq \bfo_{[h_+,+\infty)}\cEmq \bfo_{(-\infty,a)})(x),
\\\label{jointtV1m2}
 \tV^-_1(h_-,h_+;q,x) &=&q^{-1}(\cEmq \bfo_{(-\infty,h_-]}\cEpq \bfo_{(-\infty,a)})(x)\\\nonumber
 &=&
 q^{-1}(\cEmq \bfo_{(-\infty, h_-]})(x)-q^{-1}(\cEmq \bfo_{(-\infty,h_-]}\cEpq \bfo_{[a,+\infty)})(x).
 \eqa
For $\xi\in \cL^-$,
  \beqast
  \htV^+_1(h_-,h_+; q,\xi)&=&q^{-1}\phipq(\xi)\cF_{x\to \xi}\bfo_{[h_+,+\infty)}(\cEmq \bfo_{(-\infty,a)})(x)
  \\
  &=& q^{-1}\phipq(\xi)\int_{h_+}^{+\infty}dy\, e^{-iy\xi}
  %\\&&\times 
  \frac{1}{2\pi}\int_{\cL^+}e^{i(y-a)\eta}\frac{\phi^-(q,\eta)}{-i\eta}d\eta\\
  &=&-\frac{\phipq(\xi)}{2\pi q }\int_{\cL^+}\frac{e^{-ih_+(\xi-\eta)-ia\eta}}{\eta-\xi}
  \frac{\phi^-(q,\eta)}{\eta}d\eta
   \eqast
  (we can apply Fubini's theorem because there exists $c>0$ such that $\Im(\eta-\xi)>c|\eta|$ for $\eta\in \cL^+$
  and $\xi\in \cL^-$). Simplifying,
  \bbe\label{htV1pjoint}
  \htV^+_1(h_-,h_+; q,\xi)=-\frac{\phipq(\xi)e^{-ih_+\xi}}{2\pi q }\int_{\cL^+}\frac{e^{i(h_+-a)\eta}}{\eta-\xi}
    \frac{\phi^-(q,\eta)}{\eta }d\eta,\ \xi\in \cL^-.
  \ee
 Similarly, %for $\xi\in\cL^+$,
\bbe\label{htV1mjoint}
\htV^-_1(h_-,h_+;q,\xi)=e^{-ih_-\xi}\frac{\phimq(\xi)}{-i\xi q}+
\frac{\phimq(\xi)e^{-ih_-\xi}}{2\pi q }\int_{\cL^-}\frac{e^{i(h_--a)\eta}}{\eta-\xi}
    \frac{\phi^+(q,\eta)}{\eta }d\eta,\ \xi\in\cL^+.
\ee
 Define
  \beqa\label{hWp1joint}
  \hW^+_1(h_-,h_+;q,\xi)&=&-\frac{1}{2\pi }\int_{\cL^+}\frac{e^{i(h_+-a)\eta}}{\eta-\xi}
    \frac{\phi^-(q,\eta)}{\eta}d\eta,\ \xi\in\cL^-,\\
 \label{hWm1joint}
  \hW^-_1(h_-,h_+;q,\xi)&=&-\frac{1}{i\xi}+\frac{1}{2\pi}\int_{\cL^-}\frac{e^{i(h_--a)\eta}}{\eta-\xi}
    \frac{\phi^+(q,\eta)}{\eta}d\eta,\ \xi\in\cL^+.
  \eqa
  
  \item {\em Double-touch call option}, with strike $K=e^a$, $h_-<a<h_+$. We have $G(a,x)=(e^x-e^a)_+$. 
  Using the operator form of the Wiener-Hopf factorization, we simplify \eq{tVhp1} and \eq{tVhm1}:
 \beqast\\\label{jointtV1mCall}
 \tV^-_1(h_-,h_+;q,x) &=&q^{-1}\cEmq \bfo_{(-\infty,h_-]}\cEpq G(a,x)\\
\tV^+_1(h_-,h_+;q,x) &=&q^{-1}\cEpq \bfo_{[h_+,+\infty)}\cEmq G(a,x)\\
&=&q^{-1}\cEpq \bfo_{[h_+,+\infty)}(\phimq(-i)e^x-e^a)+q^{-1}\cEpq \bfo_{[h_+,+\infty)}\cEmq (e^a-e^x)_+,
 \eqast
 and calculate, for $\xi\in \cL^+$,
 \beqast
\htV^-_1(h_-,h_+;q,\xi) &=&-q^{-1}\phimq(\xi)\int_{-\infty}^{h_-}e^{-iy\xi}\frac{1}{2\pi}\int_{\cL^-} e^{iy\eta}
\phipq(\eta)\frac{e^{(1-i\eta)a}}{\eta(\eta+i)}d\eta
\\
&=&q^{-1}\phimq(\xi)e^{-ih_-\xi}\frac{i e^a}{2\pi }\int_{\cL^-}\frac{e^{i(h_--a)\eta}\phipq(\eta)}{(\eta-\xi)\eta(\eta+i)}d\eta,\
\eqast
and for $\xi\in \cL^-$,
\[
\tV^+_1(h_-,h_+;q,x) =q^{-1}\phipq(\xi)e^{-ih_+\xi}\left(\phimq(-i)\frac{e^{h_+}}{i\xi-1}-\frac{e^a}{i\xi}-\frac{i e^a}{2\pi }\int_{\cL^+}\frac{e^{i(h_+-a)\eta}\phimq(\eta)}{(\eta-\xi)\eta(\eta+i)}d\eta\right).
\]
   \end{enumerate}
  Introduce 
   \beqa\label{hW1pcall}
   \hW^+_1(h_-,h_+,q,\xi)&=&\frac{\phimq(-i)e^{h_+}}{i\xi-1}-\frac{e^a}{i\xi}
   -\frac{i e^a}{2\pi }\int_{\cL^+}\frac{e^{i(h_+-a)\eta}\phimq(\eta)}{(\eta-\xi)\eta(\eta+i)}d\eta,
   \ \xi\in\cL^-,
   \\\label{hW1mcall}
   \hW^-_1(h_-,h_+,q,\xi)&=& \frac{i e^a}{2\pi}\int_{\cL^-}\frac{e^{i(h_--a)\eta}\phipq(\eta)}{(\eta-\xi)\eta(\eta+i)}d\eta,\ \xi\in\cL^+.
   \eqa
For $j=1,2,\ldots, $ define
   \bbe\label{def:hWj}
  \hW_j^\pm(h_-,h_+; q,\xi)=qe^{ih_\pm\xi}\phi^\pm(q,\xi)^{-1}\htV_j^\pm(h_-,h_+;q,\xi), \ \xi\in\cL^\mp,
  \ee
and note that
  \bbe\label{tVtW}
  \htV_j^\pm(h_-,h_+;q,\xi)=q^{-1}e^{-ih_\pm\xi}\phi^\pm(q,\xi)\hW^\pm_j(h_-,h_+; q,\xi), \ \xi\in\cL^\mp.
  \ee
  In cases (a)-(c) above, $\htV_1^\pm(h_-,h_+;q,\xi)$ and $\hW^\pm_1(h_-,h_+;q,\xi)$ are calculated explicitly, and \eq{def:hWj}-\eq{tVtW} hold; similarly,
 $\htV_1^\pm(h_-,h_+;q,\xi)$ and $\hW^\pm_1(h_-,h_+;q,\xi)$ can be calculated for options of other types.

 \subsubsection{Main block: the iteration procedure}\label{sss: the iteration}
 For $j=2,3,\ldots,$ $\hW^\pm_j$ are calculated inductively. 
 For $\xi\in \cL^-$, we have
  \beqast
  \hW^+_j(h_-,h_+; q,\xi)&=&qe^{ih_+\xi}\cF_{x\to \xi}\bfo_{[h_+,+\infty)}(\cE^+_q)^{-1}\tV^-_{j-1}(h_-,h_+;q,x)
  \\
  &=& e^{ih_+\xi}\int_{h_+}^{+\infty}dy\, e^{-iy\xi}
  %\\&&\times 
  \frac{1}{2\pi}\int_{\cL^+}e^{i(y-h_-)\eta}\frac{\phi^-(q,\eta)}{\phi^+(q,\eta)}\hW^-_{j-1}(h_-,h_+;q,\eta)d\eta\\
  &=&-\frac{e^{ih_+\xi}}{2\pi i}\int_{\cL^+}\frac{e^{-ih_+(\xi-\eta)-ih_-\eta}}{\eta-\xi}
  \frac{\phi^-(q,\eta)}{\phi^+(q,\eta)}\hW^-_{j-1}(h_-,h_+;q,\eta)d\eta
   \eqast
  (we can apply Fubini's theorem because there exists $c>0$ such that $\Im(\eta-\xi)>c|\eta|$ for $\eta\in \cL^+$
  and $\xi\in \cL^-$). Simplifying,
  \bbe\label{hWp}
  \hW^+_j(h_-,h_+; q,\xi)=-\frac{1}{2\pi i}\int_{\cL^+}\frac{e^{i(h_+-h_-)\eta}}{\eta-\xi}
  \frac{\phi^-(q,\eta)}{\phi^+(q,\eta)}\hW^-_{j-1}(h_-,h_+;q,\eta)d\eta,\ \xi\in\cL^-.
  \ee
  Similarly, for $\xi\in \cL^+$, we calculate
  \beqast
  \hW^-_{j}(h_-,h_+; q,\xi)&=& e^{ih_-\xi}\int_{-\infty}^{h_-}dy\, e^{-iy\xi}
  %\\&&\times 
  \frac{1}{2\pi}\int_{\cL^-}e^{i(y-h_+)\eta}\frac{\phi^+(q,\eta)}{\phi^-(q,\eta)}\hW^+_{j-1}(h_-,h_+;q,\eta)d\eta\\
  &=&\frac{e^{ih_-\xi}}{2\pi i}\int_{\cL^-}\frac{e^{-ih_-(\xi-\eta)-ih_+\eta}}{\eta-\xi}
  \frac{\phi^+(q,\eta)}{\phi^-(q,\eta)}\hW^+_{j-1}(h_-,h_+;q,\eta)d\eta. 
  \eqast
  Simplifying,
  \bbe\label{hWm}
  \hW^-_{j}(h_-,h_+; q,\xi)=\frac{1}{2\pi i}\int_{\cL^-}\frac{e^{-i(h_+-h_-)\eta}}{\eta-\xi}
  \frac{\phi^+(q,\eta)}{\phi^-(q,\eta)}\hW^+_{j-1}(h_-,h_+;q,\eta)d\eta,\ \xi\in \cL^+.
  \ee
  In the cycle $j=2,3\ldots,$ we calculate  $\hW^\pm_{j}(h_-,h_+; q,\xi)$ and partials sums of the series
   \bbe\label{Winf}
   \hW^\pm(h_-,h_+; q,\xi)=\sum_{j=1}^\infty (-1)^j \hW^\pm_j(h_-,h_+; q,\xi), \ \xi\in \cL^\mp.
   \ee
   \subsubsection{Final step}\label{sss:final_step}
      For $x\in (h_-, h_+)$, we calculate
      \bbe\label{tVqpm}
      \tV^1_\pm(lG;h_-,h_+;T,x)=\frac{1}{2\pi}\int_{\cL^\mp}e^{i(x-h_\pm)\xi}\phi^\pm_q(\xi)\hW^\pm(h_-,h_+; q,\xi)d\xi,
      \ee
      then
   \bbe\label{tVqntfin}
   V^1(lG;h_-,h_+;T,x)=\frac{1}{2\pi i}\int_{\Re q=\sg}dq\,\frac{e^{qT}}{q}(\tV^1_+(lG;h_-,h_+;T,x)+
   \tV^1_-(lG;h_-,h_+;T,x)),
   \ee
   and, finally, 
   \bbe\label{eq:double_final}
  V(G;h_-,h_+;T,x)=V_{\mathrm{euro}}(lG;T,x)+V^1(lG;h_-,h_+;T,x).
 \ee
 \begin{rem}\label{rem:convergence2}{\rm Theorem \ref{thm:convergence} implies that if 
 the Gaver-Stehfest method or GWR algorithm is used to numerically evaluate
 the Bromwich integral, then the series on the RHS of \eq{tVqpm} converge for any $h_-<x<h_+$
 and any $q>0$ used in the algorithm.
 
 If the sinh-acceleration is applied to the Bromwich integral, then we need to consider the series for complex $q$
 that appear in the process of the deformation, and prove that the series admit a bound via $C(1+|q|)^m$, for some
 $C,m>0$. 

In all likelihood, for any sinh-deformation of the contour of integration in the Bromwich integral, there exists $\De h>0$ such that
if $h_+-h_-<\De h$, then the series \eq{tVqpm} fails to converge for some $q$ of interest or all $q$. However, it follows from 
Lemma \ref{lem: boundKpmmp} below that for any sinh-deformation of the contour $\{\Re q=\sg\}$ and deformed contours $\cL^\pm$ that are
in agreement with the deformation of  $\{\Re q=\sg\}$, there exists $\De h>0$ independent of $lG$ such that if $h_+-h_->\De h$, then
the series on the RHS of \eq{Winf} converges and admits  an upper bound independent of $q$.
Hence, we apply the Laplace inversion procedures formally, and verify that $h_+-h_-$ is sufficiently large so that the method works comparing the results obtains with
different Laplace inversion algorithms and different contour deformations.   
 
 }
 \end{rem}
\subsubsection{Efficient numerical evaluation of the series $\sum_{j=1}^\infty (-1)^j \hW^\pm_j(h_-,h_+; q,\xi)$. I}\label{sss:effhWseriesI}
Define operators $\cK_{+-}(=\cK_{-+}(q; \cL^+; h_-,h_+))$ and $\cK_{-+}(=\cK_{+-}(q; \cL^-; h_-,h_+))$ by
\beqa\label{defKmp}
\cK_{-+}\hu(\xi)&=&\frac{1}{2\pi }\int_{\cL^+}\frac{e^{i(h_+-h_-)\eta}}{\eta-\xi}
  \frac{\phi^-(q,\eta)}{\phi^+(q,\eta)}\hu(\eta),\ \xi\in \cL^-,\\
  \label{defKpm}
\cK_{+-}\hu(\xi)&=&\frac{1}{2\pi }\int_{\cL^-}\frac{e^{-i(h_+-h_-)\eta}}{\eta-\xi}
  \frac{\phi^+(q,\eta)}{\phi^-(q,\eta)}\hu(\eta)d\eta,\ \xi\in \cL^+.
  \eqa
  For $\be\in \bR$, define $w_\be(\xi)=e^{\be |\xi|}$, and denote by $L^\be_\infty(\cL^\pm)$ the space of measurable functions on $\cL^\pm$ with the finite norm $\|\hu\|_{L^\be_\infty(\cL^\pm)}:=\|w_\be \hu\|_{L_\infty(\cL^\pm)}$. % Note that $\cK_{-+}$ maps $L_\infty(\cL^+)$ to a space of analytic functions on the region $U(\cL_+)$ below $\cL^+$,
%  and $\cK_{+-}$ maps $L_\infty(\cL^-)$ to a space of analytic functions on the region $U(\cL_-)$ above $\cL^-$.
   Let $\om^\pm$ be from the definition of $\cL^\pm=\cL^\pm_{\chi_{\om^{1,\pm}, b^\pm, \om^\pm}}$,
  and set $\be_\pm =-(h_+-h_-)\sin|\om^\pm|$.
\begin{lem}\label{lem:boundcKpm} Let $X$ be a SINH-regular process and $q>0$. Then,
  for any  $\eps>0$, there exists $C(=C_{\ga,\eps,q})>0$ such that
  
  \beqa\label{boundcKmp}
  |\cK_{-+}\hu(\xi)|&\le &C (1+|\xi|)^{-1} \|\hu\|_{L^{\be_++\eps}_\infty(\cL^+)},\ \Im\xi\le 0,\ \hu\in L^{\be_++\eps}_\infty(\cL^+),
  \\\label{boundcKpm}
  |\cK_{+-}\hu(\xi)|&\le &C (1+|\xi|)^{-1} \|\hu\|_{L^{\be_-+\eps}_\infty(\cL^-)},\ \Im\xi\ge 0, \ \hu\in L^{\be_-+\eps}_\infty(\cL^-).
  \eqa
  \end{lem}
  \begin{proof} Eq. \eq{boundcKmp} is immediate from the following three facts: 1) $\phi^-(q,\eta)/\phi^+(q,\eta)$ is polynomially bounded
  on $\cL^+$; 2) there exists $C>0$ s.t. for $\xi\in \{\Im\xi\le 0\}$ and $\eta\in \cL^+$,
  $|1/(\eta-\xi)|\le c(1+|\xi|+|\eta|)^{-1}$, 3) $\eta\mapsto e^{ i(h_+-h_-)\eta}$  decays as $e^{\be_+|\eta|}$
  as $(\cL^+\ni)\eta\to \infty$.
  The proof of \eq{boundcKpm} is by symmetry.   
  \end{proof}
  The following Lemma is needed if the sinh-acceleration is applied to the Bromwich integral.
  \begin{lem}\label{lem: boundKpmmp} 
  Let $X$ be a SINH-regular process satisfying the conditions in Lemma \ref{lem:cones_Brom}, and let $\cL^L$ and $\cL^\pm$ be in the domains of the $q$- and $\xi$- planes where the bound \eq{bound:two cones} holds. 
  Then 
  \begin{enumerate}[(a)]
  \item
  there exist $C_1,C_2>0$ such that, for all $q\in \cL^L$,
  \beqa\label{bound:ration_cLp}
  \frac{\phi^\mp_q(\xi)}{\phi^\pm_q(\xi)}&=&C_1 \exp[C_2(\ln(1+|\xi|))^2],\ \xi\in \cL^\pm;
 % ,\\
 % \label{bound:ration_cLm}
  %\frac{\phipq(\xi)}{\phimq(\xi)}&=&C_1 \exp[C_2(\ln(1+|\xi|))^2],\ \xi\in \cL^-;
  \eqa

  \item
  there exist $C_\pm>0$ independent of $h_+, h_-$ and $q\in \cL^L$ such that 
  \beqa\label{boundcKmp2}  
  \|\cK_{-+}:L_\infty(\cL^+)\to L_\infty(\cL^-)\|&\le&C_+\int_{\cL^+}e^{-(h_+-h_-)\sin(\om^+)|\eta|+C_2(\ln(1+|\eta|))^2}|d\eta|,
  \\\label{boundcKpm2} 
  \|\cK_{+-}:L_\infty(\cL^-)\to L_\infty(\cL^+)\|&\le &C_-\int_{\cL^-}e^{-(h_+-h_-)\sin(|\om^-|)|\eta|+C_2(\ln(1+|\eta|))^2}|d\eta|.
  \eqa
  \end{enumerate}
  \end{lem}
  \begin{proof}
  (a) We use \eq{phipq_def} and \eq{eq:rec_phimq}. Since $|1/(\eta-\xi)|\le C_3/(|\xi|+\eta|)$, where $C_3>0$ is independent of
  $\xi\in \cL^-$ and $\eta\in \cL^+$, it follows from \eq{bound:two cones}
  that \eq{bound:ration_cLp} will be proved once we derive the bound
  \bbe\label{intbound0}
  \int_1^\infty \frac{|\xi| |\ln(1+\psi(\eta)/q)|}{\eta(\eta+|\xi|)}d\eta\le C_4 (1+(\ln(1+|\xi|))^2),
  \ee
  where $C_4$ is independent of $q\in \cL^L$ and $|\xi|$. 
  Since $\cL^L$ is bounded away from 0, we may prove \eq{intbound0} replacing $|\ln(1+\psi(\eta)/q)|$ with
  $C_5(1+\ln (1+\eta))$, where $C_5$ is independent of $q$ and $\eta$. After that, we use
  \[
  \int_1^{|\xi|} \frac{|\xi| (1+\ln (1+\eta))}{\eta(\eta+|\xi|)}d\eta<\int_1^{|\xi|} \frac{(1+\ln (1+\eta))}{\eta}d\eta\le C_6(1+(\ln(1+|\xi|))^2),
  \]
  and 
  \beqast
  \int_{|\xi|^\infty} \frac{|\xi| (1+\ln (1+\eta))}{\eta(\eta+|\xi|)}d\eta&<&|\xi|\int_{|\xi|}^\infty \frac{(1+\ln (1+\eta))}{\eta^2}d\eta\le (1+\ln(1+|\xi|))+C_7|\xi|\int_{|\xi|}^\infty \frac{d\eta}{\eta^2}.
  \eqast
 % This proves \eq{boundcKmp}; the proof of \eq{boundcKpm} is by symmetry. 
  (b) follows from (a) since $|1/(\eta-\xi)|\le C_3/(|\xi|+\eta|)$, where $C_3>0$ is independent of
  $\xi\in \cL^-$ and $\eta\in \cL^+$.
     \end{proof}
  We write \eq{hWp} and \eq{hWm} as 
  \bbe\label{hWpmj}
  \hW^+_{j+1}=i\cK_{-+}\hW^-_j,\  \hW^-_{j+1}=-i\cK_{+-}\hW^+_j, j=1,2,\ldots
  \ee
  and calculate $\hW^\pm_{j+1}$ and the partial sums in the cycle in $j=1,2,\ldots, M_0$. For the choice
  of the truncation parameter $M_0$
  given the error tolerance,  see Remark \ref{rem:convergence}. This choice is made assuming that
  the total error of calculation of the individual terms is sufficiently small and can be disregarded.
  
  We calculate  $\hW^\pm_{j+1}$ at points of sinh-deformed uniform grids 
  $\xi^\mp_k=i\om^{1,\mp}+b^\mp \sinh(i\om^\mp+y^\mp_k)$, $y^\mp_k=\ze^\mp k$, $k\in \bZ$, 
  on $\cL^\mp$, truncate the grids, and 
  approximate the operators $\cK_{-+}, \cK_{+-}$  with the corresponding matrix operators. The parameters of the conformal deformations
  and corresponding changes of variables and steps $\ze^\pm$ are chosen
  as in \cite{Contrarian,EfficientLevyExtremum}. If GWR algorithm is used, then the choice is especially simple;
  if the sinh-acceleration is applied to the Bromwich integral, then the parameters of the deformations of the three contours must be in a certain agreement. See \cite{EfficientLevyExtremum} for details.
  The truncation parameters $N^\pm$ are chosen taking into account the exponential rate of decay of the kernels of the integral operators
  $\cK_{-+}$ and $\cK_{+-}$ w.r.t. the second argument, and exponential decay of $e^{i(x-h_-)\xi}$ as $\xi\to\infty$ along $\cL^+$,
  and $e^{i(x-h_+)\xi}$ as $\xi\to\infty$ along $\cL^-$. In the $y^\pm$-coordinates, the rate of decay is double-exponential,
  hence, the truncation parameters $\La^\pm=N^\pm\ze^\pm$  sufficient to satisfy a small error tolerance $\eps$
  are moderately large. As a simple rule of thumb, we suggest to choose $\La^\pm$ so that
  \bbe\label{choice_Lapm}
  \exp[b^-(x-h_+)\ka_-\sin|\om^-|e^{\La^-}]<\eps,\ \exp[b^+(h_--x)\ka_+\sin(\om^+)e^{\La_+}]<\eps,
  \ee
  where $\ka_\pm\in (0,0.5)$, e.g., $\ka_\pm=0.4$. Note that a choice of a smaller $\ka_\pm$, e.g., $\ka_\pm=0.3$,
 does not  increase  $N_\pm$ significantly but makes the prescription more reliable. 
    Keeping the notation $\cK_{-+}$ and  $\cK_{+-}$ for the matrices, we calculate the matrix elements as follows:
  \beqa\label{Kpmjik}
  \cK_{+-}&=&\frac{\ze^-b^-}{2\pi}\left[\frac{e^{-i(h_+-h_-)\xi^-_k}}{\xi^-_k-\xi^+_j}
  \cdot\frac{\phipq(\xi^-_k)}{\phimq(\xi^-_k)}\cosh(i\om^-+y^-_k)\right]_{|j|\le N^+, |k|\le N^-}, \\\label{Kmpjik}
  \cK_{-+}&=&\frac{\ze^+b^+}{2\pi}\left[\frac{e^{i(h_+-h_-)\xi^+_k}}{\xi^+_k-\xi^-_j}
  \cdot\frac{\phimq(\xi^+_k)}{\phipq(\xi^+_k)}\cosh(i\om^++y^+_k)\right]_{|j|\le N^-, |k|\le N^+}.
  \eqa
   \subsubsection{Efficient evaluation of the series $\sum_{j=1}^\infty (-1)^j \hW^\pm_j(h_-,h_+; q,\xi)$. II}\label{ss:effhWseriesII}
   If the sizes of matrices $\cK_{-+}$ and  $\cK_{+-}$, hence, 
  the sizes of matrices $\cK^+:=\cK_{-+}\cK_{+-}, \cK^-:=\cK_{+-}\cK_{-+}$ are moderate so that the inverse matrices $(I-\cK^\pm)^{-1}$
  can be efficiently calculated,  the following modification of the scheme in Sect. \ref{sss:effhWseriesI} can be used to decrease the CPU time.
 
  It follows from \eq{hWpmj} that, for $k=1,2, \ldots$, and $m=1,2$,
 %\beqa\label{evenoddhWp}
 \bbe\label{hWevenodd}
  \hW^+_{2k+m}=(\cK^-)^k\hW^+_m,\ 
  %\label{evenhWm}
  \hW^-_{2k+m}=(\cK^{+})^k\hW^-_m.
  \ee
  After $\hW^\pm_m, m=1,2,$ are calculated,
  we evaluate 
  \beqa\label{sumhWpm}
  \sum_{j=1}^\infty (-1)^j \hW^\pm_j(h_-,h_+; q,\xi)&=&\sum_{k=0}^{+\infty}(\cK^\pm)^k
  (\hW^\pm_2(h_-,h_+; q,\xi)-\hW^\pm_1(h_-,h_+; q,\xi))\\\nonumber%\label{sumhWpm2}
  &=&(I-cK^\pm)^{-1}(\hW^\pm_2(h_-,h_+; q,\xi)-\hW^\pm_1(h_-,h_+; q,\xi)).
  \eqa
  It follows from Theorem \ref{thm:convergence} that the series $\cF^{-1}\sum_{k=0}^{+\infty}(\cK^\pm)^k\cF$
  converges in the operator norm for operators acting in $L_\infty$, hence, the inverses 
  \bbe\label{invCKpm}
  (I-cK^\pm)^{-1}=\sum_{k=0}^{+\infty}(\cK^+)^k
  \ee
  are well-defined; the matrix approximation are discussed above. 
  
  \begin{rem}\label{control_accum_errors}{\rm
  An accurate theoretical control of
accumulated errors of the calculation of each integral in the iterative procedure is difficult and impractical. 
 As in the previous papers where
the conformal deformation technique was used, we suggest to control the error changing the parameters of the conformal deformations.
If the difference between the results obtained with different deformations is of the order of $10^{-m}$, $m\ge 7$, then the probability that the error of each result is of the order higher than $10^{-7}$ is negligible. 
}
\end{rem}
 
  \section{Algorithms}\label{s:algo_numer}
  \subsection{Algorithm for the double no-touch option}\label{ss:algo_ntdouble}
 Steps I-V are preliminary ones; Steps VI-X constitute the main block of the algorithm and
 are performed for each $q$ used in the numerical Laplace inversion block. The latter block is  Step XI of the algorithm.
 We formulate the algorithm assuming that the sinh-acceleration is applied to the Bromwich integral; if the Gaver-Wynn Rho algorithm is used, the modification of the first  and last steps are as described in 
 \cite{EfficientLevyExtremum} for the single barrier case. %Sect. \ref{GavWynn}.
 The general prescriptions for parameters of admissible deformations are the same as in  \cite{EfficientLevyExtremum}
 but we improve the efficiency of the scheme in  \cite{EfficientLevyExtremum} using shorter grids for the calculations in
 the main block that the grid used to evaluate the Wiener-Hopf factors; in \cite{EfficientLevyExtremum}, the same long grids are used for all purposes.

  \begin{enumerate}[Step I.]
\item 
{\em Grid for the Bromwich integral.} Choose the sinh-deformation  and grid for the simplified trapezoid rule: $\vec{y}:=\ze_\ell*(0:1:N_\ell)$,
$\vec{q}:=\sg_\ell+i*b_\ell*\sinh(i*\om_\ell+\vec{y})$.
%Note that the choice must depend on $T$ but can be independent of $x_1,x_2, a_1,a_2$, at some loss in the efficiency of the algorithm.
Calculate  the derivative $\vec{der_\ell}:=i*b_\ell*\cosh(i*\om_\ell+\vec{y})$.
\item
{\em Grids for the approximations of $\hW^\pm_m$ and operators $\cK_{-+}, \cK_{+-}, \cK^\pm$}.  Choose  the sinh-deformations and grids for the simplified trapezoid rule on $\cL^\pm$: $\vec{y^\pm}:=\ze^\pm*(-N^\pm:1:N^\pm)$,
$\vec{\xi^\pm}:=i*\om_1^\pm+ b^\pm*\sinh(i*\om^\pm+i\vec{y^\pm})$.  Calculate $\psi^\pm:=\psi(\vec{\xi^\pm})$ and 
$\vec{der^\pm}:=b^\pm*\cosh(i*\om^\pm+\vec{y^\pm}).
$
\item
{\em Grids for evaluation of the Wiener-Hopf factors $\phi^\pm_q(\xi)$.} Choose longer and finer grids for the simplified trapezoid rule on $\cL^\pm_1$: $\vec{y^\pm_1}=\ze_1^\pm*(-N^\pm_1:1:N^\pm_1)$,
$\vec{\xi^\pm_1}:=i*\om^\pm_1+ b^\pm_1*\sinh(i*\om_1^\pm+i\vec{y^\pm_1})$.  Calculate $\psi^\pm_1:=\psi(\vec{\xi^\pm_1})$ and 
$\vec{der^\pm_1}:=b^\pm_1*\cosh(i*\om_1^\pm+\vec{y^\pm_1}).
$
Note that it is unnecessary to use sinh-deformations different from the ones on Step II but it is advisable to write a program allowing
for different deformations in order to be able to control errors of each block of the program separately.
\item
Calculate 2D arrays 
\beqast
D^{+-}_1&:=&1./(\mathrm{conj}(\vec{\xi^-_1})'*\mathrm{ones}(1,2*N^++1)-\mathrm{ones}(2*N^-_1+1,1)*\vec{\xi^+})),\\
D^{-+}_1&:=&1./(\mathrm{conj}(\vec{\xi^+_1})'*\mathrm{ones}(1,2*N^-+1)-\mathrm{ones}(2*N^+_1+1,1)*\vec{\xi^-})),\\
D^{+-}&:=&1./(\mathrm{conj}(\vec{\xi^-})'*\mathrm{ones}(1,2*N^++1)-\mathrm{ones}(2*N^-+1,1)*\vec{\xi^+})),\\
D^{-+}&:=&1./(\mathrm{conj}(\vec{\xi^+})'*\mathrm{ones}(1,2*N^-+1)-\mathrm{ones}(2*N^++1,1)*\vec{\xi^-})).
\eqast
\item
{\em Calculate  $\hW^+_1=-i./\vec{\xi^-}, \hW^-_1=i./\vec{\xi^+}$.}
\item
{\em Calculate} $\vec{\phipq}=\phipq(\vec{\xi^+})$ and $\vec{\phimq}=\phimq(\vec{\xi^-})$:
\beqast
\vec{\phipq}&:=&\exp((\ze^-_1*i/(2*\pi))*\vec{\xi^+_1}.*((\ln(1+\psi^-_1/q)./\vec{\xi^-_1}.*\vec{der^-_1})*D^{-+}_1)),\\
\vec{\phimq}&:=&\exp(-(\ze^+_1*i/(2*\pi))*\vec{\xi^-_1}.*((\ln(1+\psi^+_1/q)./\vec{\xi^+_1}.*\vec{der^+_1})*D^{+-}_1)),
\eqast
and then
\bbe\label{phimppm}
\phipq(\vec{\xi^-}):=1./(1+\psi^-/q)./\vec{\phimq},\ \phimq(\vec{\xi^+}):=1./(1+\psi^+/q)./\vec{\phipq}.
\ee
\item
{\em Calculate   the ratios}
\beqast
\phimq(\vec{\xi^+})./\phipq(\vec{\xi^+})&:=&1./(1+\psi^+/q)./\vec{\phipq}.^2;\\
\phipq(\vec{\xi^-})./\phimq(\vec{\xi^-})&=&1./(1+\psi^-/q)./\vec{\phimq}.^2.
\eqast
\item
{\em Calculate matrices $\cK_{+-}, \cK_{+-}$:}
\beqast
\cK_{-+}&:=&(\ze^+/(2*\pi))*\mathrm{diag}(\vec{der^+}.*(\phimq(\vec{\xi^+})./\phipq(\vec{\xi^+})).*\exp(i*(h_+-h_-)*\vec{\xi^+}))*D^{-+};\\
\cK_{+-}&:=&(\ze^-/(2*\pi))*\mathrm{diag}(\vec{der^-}.*(\phipq(\vec{\xi^-})./\phimq(\vec{\xi^-})).*\exp(-i*(h_+-h_-)*\vec{\xi^-}))*D^{+-}.\\
%\cK^-&:=&\cK_{-+}*\cK_{+-};\\
%\cK^+&:=&\cK_{+-}*\cK_{-+};
\eqast
\item
Use one of the following three blocks. Calculations using Block (1) and either Block (2) or Block (3) can be used to check the accuracy of the result.
\begin{enumerate}[(1)]
\item\begin{itemize}
\item
Assign $\hU^\pm_1=-\hW^\pm_1,$ $\hW^\pm:=\hU^\pm_1$.
\item
In the cycle $j=1,2,\ldots,M_0$, calculate
$
 \hU^+_{2}:=-i*\hU^-_1*\cK_{-+},\\ \hU^-_{2}:=i*\hU^+_1*\cK_{+-},\
 \hU^\pm:=\hW^\pm+\hU^\pm_2, \ \hU^\pm_1:=\hU^\pm_2$.
 \end{itemize}
\item
Calculate 
\begin{itemize}
\item
$\cK^-:=\cK_{-+}*\cK_{+-}, 
\cK^+:=\cK_{+-}*\cK_{-+}$;
\item
 $\hW^+_2:=i*\hW^-_1*\cK_{-+}, \hW^-_2:=-i*\hW^+_1*\cK_{+-}$;
 \item
 $
\hW^{\pm,0}:=\hW^\pm_2-\hW^\pm_1$; 
 \item
 inverse matrices 
$(I-\cK^\pm)^{-1}$;
\item
$\hW^\pm=\hW^{\pm,0}*(I-\cK^\mp)^{-1}$.
\end{itemize}
\item
Replace the last two steps of Block (2) with 
\[
\hW^\pm=\mathrm{conj}(\mathrm{linsolv}(\mathrm{diag}(\mathrm{ones}(2*N^\mp+1))-\mathrm{conj}(\cK^\mp)', \mathrm{conj}(\hW^{\pm,0})'))'.
\]
Typically, the program with  Block (1) achieves precision of the order of E-15 with $M_0=9$ or even $M_0=8$.
The CPU time is several times smaller than with Block (2). Program with Block (2) is faster if
the digitals or vanillas for many strikes need to be calculated. Block (3) is twice faster than Block (2) 
if applied only once.
\end{enumerate}
\item
{\em For $x\in (h_-,h_+)$, calculate} %$\vec{x}\subset (h_-,h_+)$, calculate}%
\beqast
 V^+&=&(\ze^-/(2*\pi))*\sum(\hW^+.*\exp(i*(x-h_+)*\vec{\xi^-}).*\phipq(\vec{\xi^-}).*\vec{der^-}), \\
V^-&=&(\ze^+/(2*\pi))*\sum(\hW^-.*\exp(i*(x-h_-)*\vec{\xi^+}).*\phimq(\vec{\xi^+}).*\vec{der^+});
\eqast
 this step can be easily parallelized for a given array $\{x_j\}$.
 \item
 {\sc Laplace inversion.} Set $Int(\vec{q})=(V^++V^-)./\vec{q}$, $Int(q_1)=Int(q_1)/2$, and,
 using the symmetry $\overline{\tV^1(q)}=\tV^1(\bar q)$, calculate
 \[
V^1=(\ze_{\ell}b_\ell/\pi)*\mathrm{real}(\mathrm{sum}(\exp(T*\vec{q}).*Int(\vec{q}).*\cosh(i*\om_\ell+\vec{y})). \]
\item
{\em Final step.} Set $V=1+V^1$.
\end{enumerate}

   \subsection{Algorithm for the double barrier digital or joint cpdf of $X_T, \barX_T, \uX_T$}\label{ss:algo_3double}
We assume that $h_-<a<h_+$. We need to make the following changes   in the algorithm for the double no-touch option:
\begin{enumerate}[(1)]
\item
Calculate the price of the digital option $V_{\mathrm{dig}}(a;T,x)$ with the payoff $\bfo_{(-\infty,a]}(x+X_T)$ at maturity
using the scheme in \cite{SINHregular}: set $x'=x-a+\mu T$, and 
\begin{itemize}
\item
if $x'\ge 0$, apply the sinh-change of variables and simplified trapezoid rule to
\[
V_{\mathrm{dig}}(a;T,x)=(2\pi)^{-1}\int_{\cL^+} \frac{e^{ix'\xi-\psi^0(\xi)}}{-i\xi}d\xi;
\]
\item
if $x'< 0$, apply the sinh-change of variables and simplified trapezoid rule to
\[
V_{\mathrm{dig}}(a;T,x)=1+(2\pi)^{-1}\int_{\cL^-} \frac{e^{ix'\xi-\psi^0(\xi)}}{-i\xi}d\xi.
\]
\end{itemize}
\item
Contrary to the case of double no-touch option, $\hW^\pm_1$ depend on $q$, and are expressed in terms of the Wiener-Hopf factors.
Therefore, Step V becomes Step VIII, and Steps VI-VIII become Steps V-VII. At Step VIII, we calculate
$\hW^\pm_1$ using \eq{phimppm} and 
$D^{-+}$ and $D^{+-}$:
\beqast
\hW^+_1&:=&-(\ze^+/(2*\pi))(\vec{der^+_1}.*\phimq(\xi^+)./\xi^+.*\exp(i*(h_+-a)*\xi^+))*D^{-+};\\
 \hW^-_1&:=&i./\xi^++(\ze^-/(2*\pi))*(\vec{der^-_1}.*\phipq(\xi^-)./(\xi^-).*\exp(i*(h_--a)*\xi^-))*D^{+-}.
\eqast  
\item
At the final step, set $V(G_a;h_-,h_+;T,x)=V_{\mathrm{dig}}(a;T,x)+V^1(G_a;h_-,h_+;T,x)$.

\end{enumerate}

 \subsection{Algorithm for the double barrier call option}\label{ss:algo_double_barr_call}
We assume that $h_-<a<h_+$. The changes are evident modifications of the changes in Sect. \ref{ss:algo_3double}:
the payoff function $G_a(x)=\bfo_{(-\infty,a]}(x)$ is replaced with $G_a=(e^x-e^a)_+$.
\begin{enumerate}[(1)]
\item
Calculate the price of the call option $V_{\mathrm{call}}(a;T,x)$ with the payoff $(e^{x+X_T}-e^a)_+$ at maturity
using the scheme in \cite{SINHregular}: set $x'=x-a+\mu T$, and 
\begin{itemize}
\item
if $x'\ge 0$, apply the sinh-change of variables and simplified trapezoid rule to
\[
V_{\mathrm{call}}(a;T,x)=-\frac{e^a}{2\pi}\int_{\cL^+} \frac{e^{ix'\xi-\psi^0(\xi)}}{\xi(\xi+i)}d\xi;
\]
\item
if $x'< 0$, apply the sinh-change of variables and simplified trapezoid rule to
\[
V_{\mathrm{call}}(a;T,x)=e^{x+T*\psi(-i)}-e^a-\frac{e^a}{2\pi}\int_{\cL^-} \frac{e^{ix'\xi-\psi^0(\xi)}}{\xi(\xi+i)}d\xi.
\]
\end{itemize}
\item
Contrary to the case of double no-touch option, $\hW^\pm_1$ depend on $q$, and are expressed in terms of the Wiener-Hopf factors.
Therefore, Step V becomes Step VIII, and Steps VI-VIII become Steps V-VII. On Step VIII, we 
 calculate
$\hW^\pm_1$ using \eq{phimppm} and $D^{-+}$ and $D^{+-}$:

\beqast
   \hW^+_1&:=&(\phimq(-i)*\exp(h_+))./(i*{\xi^-}-1)+\exp(a)./(-i*{\xi^-}) \\
   &&-(i*\ze^+*\exp(a)/(2*\pi))*(\exp(i*(h_+-a)*{\xi^+}).*\phipq({\xi^+})./{\xi^+}./({\xi^+}+i))*D^{-+},\\
  \hW^-_1&:=&(i*\ze^-*\exp(a)/(2*\pi))*(\exp(i*(h_--a)*{\xi^-}).*\phipq({\xi^-})./{\xi^-}./({\xi^-}+i))*D^{+-}.
   \eqast
\item
At the final step, set $V(G_a;h_-,h_+;T,x)=V_{\mathrm{call}}(a;T,x)+V^1(G_a;h_-,h_+;T,x)$.
\end{enumerate}

\section{Numerical examples}\label{s:numer}
\subsection{General remarks}\label{ss:numer_gen_rem}
The calculations in the paper
were performed in MATLAB 2017b-academic use, on a MacPro Chip Apple M1 Max Pro chip
with 10-core CPU, 24-core GPU, 16-core Neural Engine 32GB unified memory,
1TB SSD storage.
The  CPU times shown can be significantly improved using parallelized calculations of the Wiener-Hopf factors and the main block, for each $q$ used
in the Laplace inversion procedure, especially if the sinh-acceleration is applied to the Bromwich integral. The parallelization w.r.t.
$h_\pm, x, a$ is also possible.
As in the numerical examples in \cite{EfficientLevyExtremum},  we use a KoBoL with the characteristic exponent 
$
\psi(\xi)=-i\mu\xi+\psi^0(\xi),$ %c\Gamma(-\nu)(\lp^\nu-(\lp+i\xi)^\nu+(-\lm)^\nu-(-\lm-i\xi)^\nu), \]
where $\psi^0$ is given by \eq{KBLnuneq01} with
$\lp=1,\lm=-2$ and
(I) $\nu=0.2$, hence, the process is close to Variance Gamma; (II) 
$\nu=1.2$, hence, the process is close to NIG, and of infinite variation. In addition, we include several examples with $\nu=0.8$:
the process is close to NIG but of finite variation.
In the majority of examples, $\mu=0$, which allows us to apply the sinh-acceleration to the Bromwich integral and efficiently control the
errors not only in the case $\nu\ge 1$ but in the case $\nu<1$ as well. 
In all examples, $c>0$ is chosen so that the second instantaneous moment $m_2=\psi^{\prime\prime}(0)=0.1$; the riskless rate $r$ is chosen from the no-arbitrage condition $r+\psi(-i)=0$.
The same algorithms
will produce results of similar efficiency for Stieltjest-L\'evy processes of the same type and order.  If the order is 0+ (VGP) or very close to 0,
the algorithms may require
much longer grids, hence, larger CPU time, to achieve the precision shown in our examples. Serious difficulties 
arise if the order is $\nu=1+$ or $\nu$ is very close to 1, and the L\'evy density is strongly asymmetric near 0 (e.g., 
KoBoL of order $\nu=1$ or close to 1 with $c_+\neq c_-$).
See related examples in \cite{ConfAccelerationStable} for stable L\'evy distributions. 
Longer arrays are needed if the underlying is very close to one of the barriers, which explains  rather large CPU times in Examples 
\ref{ex:NT4999} and \ref{ex:VdoubleDig}. Naturally, accurate calculations are especially difficult if the distance between
the barriers is very small, and the steepness parameters $\la_\pm$ are small in absolute value. We present examples for a moderately small distance $h_+-h_-=0.1$, and moderately small $\lp=1$, $\lm=-2$.

We consider double barrier no-touch options, digitals (equivalently,
the joint distribution of the process and its supremum and infimum processes), and calls. 
The time to maturity is small $T=0.004$, moderately small, $T=0.25$, moderate $T=1$, and moderately large: $T=3$ and $T=5$. In our examples, the distance between
barriers is rather small, hence, even at $T=1$ the prices are very small, and negligible at $T=3, 5$. Nevertheless, the relative error is rather small
even in the case of prices of the order of $10^{-12}$. 
We also produce the graphs of option prices not extremely close to maturity which demonstrate the apparent difficulties
of calculation of prices of double barrier options using time discretization (Carr's randomization) and interpolation of option prices at each time step (see \cite{BLdouble} and the bibliography therein).
We leave to the future the systematic study of accuracy of methods based on time discretization and calculations in the state space.
Numerical examples in \cite{EfficientDiscExtremum} show that, in the case of options with one barrier, a moderately good accuracy can be achieved using time discretization provided the calculations are in the dual space.

When the sinh-acceleration is applied to the Bromwich integral,
we can achieve the precision of the order of E-15 and better; in the same cases, if the GWR acceleration is applied,
the smallest errors are in the range E-11 to E-5; we surmise that the errors of the results obtained with the GWR acceleration in the case $\nu<1$ and $\mu\neq 0$ are of the same order. Differences between results obtained with GWR and different deformations of the contours of integration
in the formulas for the Laplace transform for $q>0$ are smaller, in some cases, by a factor of 10 and more (to save space, we
do not include the tables with these results). The benchmark values and other values in the tables are produced with $M_0=9$ using two sets of deformations of the contours of
integration and the algorithm with Step IX(1). We check the accuracy of the results running the program with Step IX (2), 
which is slower. For large maturities ($T=3$ in the case $\nu=1.2$ and $T=5$ in the case $\nu=0.2$), only the algorithm with Step IX(2) (or IX(3)) produces good results.  In the tables, we show the sizes of grids which do not lead to the increase of the order of errors of prices relative to the benchmark; by a more systematic effort, one can find arrays of smaller sizes, that produce results with the same precision.

Numerical experiments confirm the observation that we made in our previous publications
\cite{paraLaplace,paired,Contrarian,EfficientLevyExtremum} about the efficiency of the GWR method. Namely, the precision of the final result of the order of $10^{-6}-10^{-8}$ can be achieved if the Laplace transform is calculated with the precision of 
the order of $10^{-12}-10^{-14}$ (the general recommendation is the precision of the order of $10^{2.2M}=10^{-17}$, where $2M$ is the
number of terms in the GWR algorithm; we use $M=8$). 
The second observation which we made in \cite{SINHregular} in the context of applications of the sinh-acceleration to pricing European options is that,
contrary to the Fourier transform methods that do not use conformal deformations of the contours of integration,
calculations are more efficient in the case of processes close (but not very close) to the Variance Gamma (VG) model.
There are several conceptual explanations to this effect. The first one is the same as in the case of European options, namely,
if the order of the process $\nu<1$, the strips of analyticity in the new coordinates are wider than in the case of process of infinite variation,
hence, the step in the infinite trapezoid rule can be chosen larger. At the same time, in the new coordinates, the rates of decay of the integrands
 are approximately the same for all processes unless the process is VG or very close to VG.
 The second explanation follows from the analysis of the behavior of the price near the barrier (see \cite{NG-MBS,BIL}).
 For processes of infinite variation, the derivative of the price tends to infinity as the underlying approaches the barrier
 faster than for processes of finite variation. Since the last step of the algorithm is the Fourier inversion, the
 irregularity of the price requires the use of longer and finer grids. The third reason is that we compare the prices in two models with
 the same second instantaneous moment and time to maturity. If $\nu=0.2$,  very small jumps dominate, hence,  the probability of the process hitting one of the barriers is smaller, hence, the boundary effects are smaller as well.

\subsection{Examples of calculations at many points}\label{ss:numer_many}
The first two example show that even without parallelization w.r.t. $x$ and $a$, the calculation of the price at thousands points with the precision of the order of E-14 is possible in a dozen of seconds. 
%Without parallelization, the CPU time is approximately proportional to the number of points. 
\begin{example}\label{ex:NT4999}{\em In Fig. \ref{fig:DoubleNT4999_2}, we show the graph of a no-touch option as a function of
$x\in [-0.04998,0.04998]$, at 4,999 points. $X$ is KoBoL of order $\nu=1.2$, $\mu=0$, 
$\lp=1,\lm=-2$, $c=c_+=c_-$ is defined from $m_2:=\psi^{\prime\prime}(0)=0.1$. 
Precision E-14 is achieved in 11.0 sec.  if the sinh-acceleration is used.

The upper panel: graph for $x\in (h_-,h_+)$. The middle panel: graph in a small vicinity of $h_-$.
The lower panel: price shown in the middle panel is normalized by $(x-h_-)^{\nu/2}$. As it is proved in \cite{BIL} in the case of
single barrier options, in the case $\nu>1$ and  in the case $\nu\in (0,1), \mu=0$, the price $V(x)$ has the asymptotics $V(x)\sim A(x-h_-)^{\nu/2}$
as $x\downarrow h_-$.
Using the standard localization results for boundary problems for elliptic \cite{eskin} and quasi-elliptic (in particular, parabolic)
pseudo-differential operators \cite{DegEllEq,LevPan}, it is possible to prove that the asymptotics of the same form is valid for barrier options with two barriers (although, contrary to \cite{BIL}, an explicit formula for the asymptotic coefficient $A$ is impossible to derive). 
Numerical examples in \cite{BIL} demonstrate that the asymptotic formula is moderately accurate in a very small vicinity of
the barrier only, and the quality of approximation decreases as $\nu$ approaches either 0 or 1. The presence of the second barrier 
decreases the accuracy of the asymptotic approximation further still but, as the lower panel demonstrates, the relative error of the asymptotic formula is of the order of  10\%.
}
\end{example}

\begin{example}\label{ex:VdoubleDig}{\em In Fig. \ref{fig:VdoubleDig}, we show the graph of the double barrier digital as a function of $(x,a)\in 
[-0.049,0.049]^2$, at $99^2=9,981$ points. $T=0.25, h_\pm=\pm 0.05$.
 $X$ is KoBoL of order $\nu=1.2$, $\mu=0.02$, 
$\lp=1,\lm=-2$, $c=c_+=c_-$ is defined from $\psi^{\prime\prime}(0)=0.1$. 
Precision E-14 is achieved in 229 sec, E-7 in 22.6 sec. if the sinh-acceleration is used; if the GWR acceleration is used,
then precision better than E-5 is achieved in 11.1 sec.

\subsection{Shapes of price curves}\label{ss:numer_shapes}
We plot the price curves of the double barrier no-touch option (Fig. \ref{fig:DoubleNTT001nu08mu002})
 digital (Fig. \ref{fig:DoubleDigT001nu008mu002a0}) and call option (Fig. \ref{fig:DoubleCallT001nu008mu002a0})
in the exponential L\'evy model $S_t=e^{X_t}$, close to  maturity date $T=0.01$; 
the barriers are $H_\pm=e^{h_\pm}$, the strike is $K=1=e^0$. 
$X$ is KoBoL of order $\nu=0.8$, $\mu=0.02$, 
$\lp=1,\lm=-2$, $c=c_+=c_-$ is defined from $m_2:=\psi^{\prime\prime}(0)=0.1$. 
%The riskless rate $r=0$. Since $\psi(-i)=0.02$, 
%the no-arbitrage condition $r+\psi(-i)=0$ is violated but the reader can multiply the prices by $e^{-0.0002}$ to get
%the no-arbitrage price. 
Since $\nu<1$ and $\mu>0$, $V(h_-+)>0$. This fact is proved in \cite{BIL} for single barrier options.
Using the standard localization results for boundary problems for elliptic \cite{eskin} and quasi-elliptic (in particular, parabolic)
pseudo-differential operators \cite{DegEllEq,LevPan}, it is possible to prove that $V(h_-+)>0$  for barrier options with two barriers, if $\nu<1$ and $\mu>0$.
Naturally, in the case of the call option of small maturity, and strike half-way between the barriers, the limit is very small.

On the other hand, if the process is close to VG, then, even in the case $\mu=0$, the price may seem to have a positive limit at $h_-+$
because for processes close to VG, the impact of large jumps is comparable to the impact of small fluctuations,
hence,  the Blumenthal 0-1 law manifests itself only very close to the boundary. 
On Fig. \ref{fig:Vntnu02mu0T025}-\ref{fig:Vntnu02mu0T025hm2}, we plot the graph of the call option price. The parameters are $T=0.25, 
K=e^0=1, \nu=0.2, m_2=0.1, \lp=1,\lm=-2, \mu=0$, hence, $V(h_-+0)=0$.
On Fig. \ref{fig:Vntnu02mu0T025}, where the graph on $[h_-+0.001, h_+-0.001]$ is plotted, it seems that $V(h_-+0)>0$.
On Fig. \ref{fig:Vntnu02mu0T025hm2}, upper panel, we plot the price curve in a small vicinity of $h_-$. It is seen that the price started to decrease as
$x$ approaches $h_-$. The lower panel shows the normalized price $V_{call;norm}(x)=V_{call}(x)/(x-h_-)^{\nu/2}$. 

Clearly, if the price is calculated using the calculations in the state space, backward induction and interpolation at each time step,
then, if
the linear interpolation is used and $\De x\ge 10^{-4}$, where $\De x$ is the step of the grid used for the interpolation, 
the error introduced at each step is of the order of $\De x$. Interpolation of higher order may produce huge errors because
even the first derivative of the price tends to infinity as the underlying approaches the barrier.

\subsection{Tables: dependence of precision on the method for the Laplace inversion, time to maturity
and the order of the process}\label{ss:numer_tables} In tables in Section \ref{ss:Fig and tables}, 
the ``benchmark values" are the values obtained using the sinh-acceleration with different deformations of the three contours of integration;
in all examples, we use the algorithm with Step IX(1) (truncation of the series) and $M_0=9$.  We calculate the error of the ``benchmark values" w.r.t.  the values obtained with the scheme with step IX(2)
(using the matrix inverse instead of the truncated sum). We stop increasing $T$ when the truncation error becomes larger than E-15, and larger $M_0$ are necessary to use to satisfy the error tolerance E=15. If $\nu=1.2$, this happens at $T=1$, if $\nu=0.2$, at $T=3$. We determine where to stop looking at the prices of the call option.
If $\nu=0.2$, then the scheme with Step IX(2) achieves accuracy of the order of E-11 at $T=5$; at $T=10$,  high precision arithmetic is needed. If $\nu=1.2$, the scheme with Step IX(2) achieves the precision better than E-16 at $T=3$; since the prices are of the order of $E-12$, the relative error is smaller than $E-03$.
For $T\le 1$,  the differences are of the order of E-15 or better.
If the order $\nu<1$ and drift $\mu\neq 0$, then only the GWR acceleration is applicable. 

\begin{rem}\label{rem: strange} {\em Some rows of prices may seem wrong. For instance, for $T\le 3$ in Table 6,
the sequences of prices are increasing if $T=0.004, 0.25, 3$. The reason is that the grid step is 0.01. For $T=0.004, 0.25,$ the peak of the price curve is on $(0.04,0.005)$, for $T=3$, the peak is on $(0.03,0.04)$. See Fig. \ref{fig:Vntnu02mu0T025} and \ref{fig:Vntnu02mu0T3}.
}
\end{rem} 

%We apply GWR with $M=8$, and, 
%as the benchmark, we show the differences between values obtained with two pairs of deformations of the contours used
%to evaluate the Laplace transform of the price for $q$ used in the GWR algorithm.  We show the errors obtained with different sets of %parameters $N_\ell, N^\pm, N^\pm_1$ shown in the tables. With the parametrization of the deformations that we use, the steps $\ze$  in t%he simplified trapezoid rule are in the range $\ln (1/\eps)/(2\pi)*(1.1, 1.7)$, where $\eps$ is the error. 
%\cite{EfficientLevyExtremum}, \cite{SINHregular} In \cite{EfficientAmenable} \cite{EfficientStableLevyExtremum}

\section{Conclusion}\label{concl}
In the paper, we develop a very fast and accurate method for pricing double barrier options with continuous monitoring in
wide classes of  L\'evy models; the method extends the method developed in \cite{EfficientLevyExtremum} for
single barrier options and more general options with barrier/lookback features.  For each $q$ in the Laplace inversion formula (we use the sinh-acceleration in the Bromwich integrals and the Gaver-Wynn Rho acceleration algorithm), the main block  is the evaluation of two series of perpetual first touch options (or their analytic continuation
w.r.t. $q$). The iteration procedure for the evaluation of the series in the operator form is the same as in \cite{BLdouble},
where the calculations are in the state space and the technique of the expected present value operators (EPV-operators) is used. The calculations in the dual space allow for much more efficient calculations.
The iteration procedure in the present paper can be applied in a similar vein to double barrier 
options with discrete monitoring and calculation of joint probability distributions  in stable L\'evy models. See \cite{EfficientDiscExtremum,EfficientStableLevyExtremum} for the single barrier case.
The procedure can be applied to improve the performance of pricing of barrier and American options in regime-switching L\'evy models, approximations of models with stochastic volatility and interest rates with regime-switching models, where the technique
of the EPV-operators was used \cite{ExitRSw,stoch-int-rate-CF,SVolSSRN,BLHestonStIR08,MSdouble,BarrStIR}.

 For wide regions in the parameter space, the precision of the order of $10^{-15}$ is achievable in seconds, and of the order of $10^{-9}-10^{-8}$ - in  fractions of a second. The Wiener-Hopf factors and repeated integrals in the pricing formulas
 are calculated using the sinh-deformation of the lines of integration, the corresponding changes of variables
 and the simplified trapezoid rule. If the Bromwich integral is calculated using the Gaver-Wynn Rho acceleration
 instead of the sinh-acceleration, the CPU time is typically smaller but the precision is of the order of $10^{-9}-10^{-6}$, at best.
Explicit pricing algorithms and numerical examples are for no-touch options, digitals (equivalently, for
the joint distribution function of a L\'evy process and its supremum and infimum processes), and call options.

We produce several graphs to explain fundamental difficulties for accurate pricing of barrier options using the time discretization
or Carr's randomization and interpolation-based calculations in the state space.

\appendix

\section{Figures and tables}\label{ss:Fig and tables}

\begin{table}
\caption{\small Double barrier no-touch option. KoBoL 
close to NIG, with an almost symmetric jump density, and no ``drift": $m_2=0.1, \nu=1.2, \lm=-2, \lp=1, \mu=0$. Riskless rate $r=0$.
Prices and errors (rounded) of the algorithm with the sinh- and GWR-acceleration applied to the Bromwich integral (SINH and
GWR).  (Log-)barriers: $h_-=-0.05, h_+=0.05$. Time to maturity $T=0.004, 0.25, 1$.
 }
 {\tiny
\begin{tabular}{c|ccccc}
\hline\hline
$x$ & -0.04 & -0.02  & 0 & 0.02 & 0.04\\\hline\hline
$T=0.004$ & & & &  \\\hline
$V_{nt}$ & 0.944232464403407 & 0.984791837906914 & 0.988695065999628 & 0.985130282346314 & 0.945243176095013\\
$err_1$ & -4.13E-08	&	-1.176E-08&	-9.88E-09	& 	-1.29E-08	-&	-3.951E-08\\
 $err_2$ & -1.09E-07	 &	6.22E-09	&	9.06E-09	&	1.15E-08	&	-1.12E-07
 \\\hline\hline
 $T=0.25$ & & & &  \\\hline
$V_{nt}$ & 0.0925697509133228 & 0.183597478719832 & 0.216239237263554 & 0.187081211429371 &0.0961682820257716\\
$err_1$ & 1.05E-06	&	6.87E-08	&	7.83E-08	 &	6.67E-08	&	1.03E-06\\
 $err_2$ & 8.49E-06	&	5.92E-06	&	5.63E-07	&	6.22E-06	&	9.79E-06\\\hline\hline

 $T=1$ & & & &  \\\hline
$V_{nt}$ & 0.000488706725350729 & 0.000970205697557125 & 0.00114386828643243 &0.000989805061225368 &
0.000508651147353323\\
$err_1$ & -2.67E-05	&	-9.94E-06 &	-6.30E-06 &		-1.05E-05 &		-2.81E-05\\
 $err_2$ & 1.23E-05	& 	1.39E-05 &	1.01E-05 	&	1.47E-05 &		1.27E-05\\\hline\hline
\end{tabular}
}
\begin{flushleft}{\tiny
$\underline{T=0.004}$. Benchmark values at 9 points: CPU time 15.2 sec, precision is better than E-15. \\
Sizes of arrays for the benchmark:
$N_\ell=417$, $N^-=	214$, $N^+=	241$, $N^\pm_1=	1504$.\\
$\eps_1$: errors of SINH with $N_\ell=112$, $N^-=	63$, $N^+=	84$, $N^\pm_1=298$. CPU time 0.746 sec.\\
$\eps_2$: error of GWR with $M=8$, $N^-=23$, $N^+=38$, $N^\pm_1=93$. CPU time 0.097 sec.

\vskip0.1cm
\noindent

$\underline{T=0.25}$. Benchmark values at 9 points: CPU time 6.80 sec, precision is better than E-15. \\
Sizes of arrays for the benchmark:
$N_\ell=230$, $N^-=	171$, $N^+=	182$, $N^\pm_1=	1193$.\\
$\eps_1$: errors of SINH with $N_\ell=59$, $N^-=	69$, $N^+=	79$, $N^\pm_1=203$. CPU time 0.615 sec.\\
$\eps_2$: error of GWR with $M=8$, $N^-=136$, $N^+=144$, $N^\pm_1=954$. CPU time 0.375 sec.

\vskip0.1cm
\noindent

$\underline{T=1}$. Benchmark values at 9 points: CPU time 9.17 sec, precision is better than E-15. \\
Sizes of arrays for the benchmark:
$N_\ell=215$, $N^-=	198$, $N^+=	222$, $N^\pm_1=	1676$.\\
$\eps_1$: errors of SINH with $N_\ell=29$, $N^-=	29$, $N^+=	48$, $N^\pm_1=117$. CPU time 0.111 sec.\\
$\eps_2$: error of GWR with $M=8$, $N^-=23$, $N^+=38$, $N^\pm_1=93$. CPU time 0.097 sec.

}
\end{flushleft}

\label{table1nu1.2}
 \end{table}

\begin{table}
\caption{\small Double barrier no-touch option. KoBoL 
close to VG, with an almost symmetric jump density, and no drift: $m_2=0.1, \nu=0.2, \lm=-2, \lp=1, \mu=0$. Riskless rate $r=0$.
Prices and errors (rounded) of the algorithm with the sinh- and GWR-acceleration applied to the Bromwich integral (SINH and
GWR).  (Log-)barriers: $h_-=-0.05, h_+=0.05$. Time to maturity $T=0.004, 0.25, 3$.
 }
 {\tiny
\begin{tabular}{c|ccccc}
\hline\hline
$x$ & -0.04 & -0.02  & 0 & 0.02 & 0.04\\\hline\hline
$T=0.004$ & & & &  \\\hline
$V_{nt}$ & 0.997159234166403 & 0.99785039353072 & 0.997988709856923 & 0.997873055193352 & 0.997205955464661\\
$err_1$ & 2.53E-11 &		5.77E-11	&	5.27E-11	& 3.80E-11&	6.86E-11\\
 $err_2$ & 5.13E-10	& 	-8.84E-10 & 	-5.30E-10	&	4.14E-10 &		-1.32E-11 \\\hline\hline
 
 $T=0.25$ & & & &  \\\hline  
$V_{nt}$ &    
0.837255746301533 & 0.872998705974284 & 0.880407965481731 & 0.87420738492834 &
 0.83967624398896\\
$err_1$ & 5.16E-08 &		5.60E-08	&	-5.45E-11	&	2.37E-08 &		2.23E-08\\
 $err_2$ & 2.43E-08	&	3.31E-08 &	2.88E-08	&		2.75E-08	&	2.61E-08\\\hline\hline

 $T=3$ & & & &  \\\hline
$V_{nt}$ & 0.133264677579268 & 0.179416477579805 & 0.192359856627979 &
0.181619873768539 & 0.136797832249264
\\
$err_1$ & -3.33E-08	&	-2.82E-08	&	-2.65E-08 &		-2.70E-08	&	-3.34E-08\\
$err_2$ & 2.48E-08	&	5.42E-08	&	2.05E-08	&	1.49E-08	&	2.25E-09\\

 1.23E-05	& 	1.39E-05 &	1.01E-05 	&	1.47E-05 &		1.27E-05\\\hline\hline
\end{tabular}
}
\begin{flushleft}{\tiny
$\underline{T=0.004}$. Benchmark values at 9 points: CPU time 9.83 sec, precision is better than E-15. \\
Sizes of arrays for the benchmark:
$N_\ell=356$, $N^-=	164$, $N^+=	174$, $N^\pm_1=	1254$.\\
$\eps_1$: errors of SINH with $N_\ell=171$, $N^-=	73$, $N^+=	92$, $N^\pm_1=379$. CPU time 2.037 sec.\\
$\eps_2$: error of GWR with $M=8$, $N^-=59$, $N^+=75$, $N^\pm_1=311$. CPU time 0.171 sec.
\vskip0.1cm
\noindent

$\underline{T=0.25}$. Benchmark values at 9 points: CPU time 6.50 sec, precision is better than E-15. \\
Sizes of arrays for the benchmark:
$N_\ell=230$, $N^-=	177$, $N^+=	188$, $N^\pm_1=	979$.\\
$\eps_1$: errors of SINH with $N_\ell=72$, $N^-=	48$, $N^+=	67$, $N^\pm_1=162$. CPU time 0.377 sec.\\
$\eps_2$: error of GWR with $M=8$, $N^-=136$, $N^+=144$, $N^\pm_1=954$. CPU time 0.402 sec.

\vskip0.1cm
\noindent

$\underline{T=3}$. Benchmark values at 9 points: CPU time 4.45 sec, precision is better than E-15. \\
Sizes of arrays for the benchmark:
$N_\ell=152$, $N^-=	164$, $N^+=	184$, $N^\pm_1=	1397$.\\
$\eps_1$: errors of SINH with 
$N_\ell=34$, $N^-=	47$, $N^+=	60$, $N^\pm_1=190$. CPU time 0.204 sec.\\
$\eps_2$: error of GWR with $M=8$, $N^-=30$, $N^+=40$, $N^\pm_1=112$. CPU time 0.119 sec.

}
\end{flushleft}

\label{table1nu.02}
 \end{table}

\begin{table}
\caption{\small Double barrier digital put. (Log-)barriers: $h_-=-0.05, h_+=0.05$; log-strike $a=-0.01$.
KoBoL 
close to NIG, with an almost symmetric jump density, and no ``drift": $m_2=0.1, \nu=1.2, \lm=-2, \lp=1, \mu=0$. 
Riskless rate $r=0$.
Prices and errors (rounded) of the algorithm with the sinh- and GWR-acceleration applied to the Bromwich integral (SINH and
GWR).  Time to maturity $T=0.004, 0.25, 1$.
 }
 {\tiny
\begin{tabular}{c|ccccc}
\hline\hline
$x$ & -0.04 & -0.02  & 0 & 0.02 & 0.04\\\hline\hline
\\\hline\hline
$T=0.004$ & & & &  \\\hline
$V_{dig}$ & 0.936033131420221 & 0.942743923266939 & 0.0407165015135701 &
0.00756840253469884 & 0.00309395728748227\\
$\eps_1$ & 6.64E-10 &		-3.51E-10	&	3.28E-10	&	5.08E-10	&	-6.12E-10\\
$\eps_2$ & -6.13E-08&	5.34E-09	&	2.57E-09	&	8.58E-10	&	2.05E-09
\\\hline\hline
$T=0.25$ & & & &  \\\hline
$V_{dig}$ & 0.0367167500936684 & 0.0706406957622098 & 0.0786094461300362 & 0.0641355992734415 & 0.0319013015638287\\
$\eps_1$ & 1.48E-07 &		3.23E-08 &		1.88E-08 &		1.36E-08 &	1.08E-08\\
$\eps_2$ & -9.39E-07 &1.27E-06 & -6.34E-06 &-1.97E-06 &-1.83E-06
\\\hline\hline

$T=1$ & & & &  \\\hline
$V_{dig}$ & 0.000177771957381001 & 0.000352921137775741 & 0.000416090671743419 &
0.000360047453580148 & 0.000185024416500812\\
$\eps_1$ & 1.07E-07	&	1.32E-07	&	6.97E-08	&	4.12E-08	&	2.85E-08\\
$\eps_2$ & -1.73E-05&	-6.21E-06	&	8.72E-07	&	9.36E-06	&	8.59E-06\\

\end{tabular}
}
\begin{flushleft}{\tiny
$\underline{T=0.004}$. Benchmark values at 9 points: CPU time 13.8 sec, precision is better than E-15. \\
Sizes of arrays for the benchmark:
$N_\ell=449$, $N^-=	177$, $N^+=	197$, $N^\pm_1=	1088$.\\

$\eps_1$: errors of SINH with $N_\ell=158$, $N^-=	49$, $N^+=	76$, $N^\pm_1=184$. CPU time 1.00 sec.\\
$\eps_2$: error of GWR with $M=8$, $N^-=82$, $N^+=99$, $N^\pm_1=387$. CPU time 0.237 sec.
\vskip0.1cm
\noindent
$\underline{T=0.25}$. Benchmark values at 9 points: CPU time 8.81 sec, precision is better than E-15. \\
Sizes of arrays for the benchmark:
$N_\ell=305$, $N^-=	167$, $N^+=	178$, $N^\pm_1=	1891$.\\

$\eps_1$: errors of SINH with $N_\ell=86$, $N^-=	46$, $N^+=	65$, $N^\pm_1=152$. CPU time 0.413 sec.\\
$\eps_2$: error of GWR with $M=8$, $N^-=136$, $N^+=144$, $N^\pm_1=731$. CPU time 0.399 sec.
\vskip0.1cm
\noindent
$\underline{T=1}$. Benchmark values at 9 points: CPU time 9.29 sec, precision is better than E-15. \\
Sizes of arrays for the benchmark:
$N_\ell=259$, $N^-=	200$, $N^+=	241$, $N^\pm_1=	1410$.\\

$\eps_1$: errors of SINH with $N_\ell=51$, $N^-=	43$, $N^+=	69$, $N^\pm_1=157$. CPU time 0.267 sec.\\
$\eps_2$: error of GWR with $M=8$, $N^-=84$, $N^+=110$, $N^\pm_1=422$. CPU time 0.273 sec.

}
\end{flushleft}

\label{table2nu1.2}
\end{table}

\begin{table}
\caption{\small Double barrier digital put. (Log-)barriers: $h_-=-0.05, h_+=0.05$; log-strike $a=-0.01$.
KoBoL 
close to VG, with an almost symmetric jump density, and no ``drift": $m_2=0.1, \nu=0.2, \lm=-2, \lp=1, \mu=0$. 
Riskless rate $r=0$.
Prices and errors (rounded) of the algorithm with the sinh- and GWR-acceleration applied to the Bromwich integral (SINH and
GWR).  Time to maturity $T=0.004, 0.25, 3$.
 }
 {\tiny
\begin{tabular}{c|ccccc}
\hline\hline
$x$ & -0.04 & -0.02  & 0 & 0.02 & 0.04\\\hline\hline

$T=0.004$ & & & &  \\\hline
$V_{dig}$ & 0.996564411171869 & 0.99657921660561 & 0.00112224938146623 &
0.000499111365249944 & 0.00031397891513379\\
$\eps_1$ & 6.64E-10 &		-3.51E-10	&	3.28E-10	&	5.08E-10	&	-6.12E-10\\
$\eps_2$ & -1.81E-10 &	 -1.23E-10	 &	5.16E-11&	2.89E-11&	-3.29E-14
\\\hline\hline
$T=0.25$ & & & &  \\\hline
$V_{dig}$ & 0.806048752314656 & 0.808342339586413 & 0.0564741685789332 & 0.0263104606596591 &
0.0165845924903516\\
$\eps_1$ & 3.56E-08	&	2.30E-08	&	1.86E-08	&	1.53E-08 &		1.28E-08\\
$\eps_2$ & 3.57E-08	&	1.618E-08	&	7.169E-09	& 3.163E-08	&	2.758E-08
\\\hline\hline

$T=3$ & & & &  \\\hline
$V_{dig}$ & 0.0826482000571708 & 0.0963948765801401 & 0.064756560941078 & 0.0445502961594203 
& 0.0292378743031199\\

$\eps_1$ & 7.40E-08&	4.54E-08	&	4.136E-08	 &	1.99E-08& 3.347E-08\\
$\eps_2$ & -9.69E-07 &		4.54E-07	&	8.08E-08	&	-3.42E-07	&	-1.72E-08\\

\end{tabular}
}
\begin{flushleft}{\tiny
$\underline{T=0.004}$. Benchmark values at 9 points: CPU time 6.92 sec, precision is better than E-15. \\
Sizes of arrays for the benchmark:
$N_\ell=230$, $N^-=	165$, $N^+=	199$, $N^\pm_1=	1175$.\\
$\eps_1$: errors of SINH with $N_\ell=143$, $N^-=	51$, $N^+=	69$, $N^\pm_1=178$. CPU time 0.861sec.\\
$\eps_2$: error of GWR with $M=8$, $N^-=69$, $N^+=91$, $N^\pm_1=362$. CPU time 0.290 sec.
\vskip0.1cm
\noindent
$\underline{T=0.25}$. Benchmark values at 9 points: CPU time 8.81 sec, precision is better than E-15. \\
Sizes of arrays for the benchmark:
$N_\ell=230$, $N^-=	165$, $N^+=	199$, $N^\pm_1=	117$.\\

$\eps_1$: errors of SINH with $N_\ell=72$, $N^-=	50$, $N^+=	75$, $N^\pm_1=193$. CPU time 0.414 sec.\\
$\eps_2$: error of GWR with $M=8$, $N^-=32$, $N^+=51$, $N^\pm_1=731$. CPU time 0.075 sec.

\vskip0.1cm
\noindent
$\underline{T=3}$. Benchmark values at 9 points: CPU time 5.56 sec, precision is better than E-15. \\
Sizes of arrays for the benchmark:
$N_\ell=164$, $N^-=	189$, $N^+=	213$, $N^\pm_1=	1950$.\\

$\eps_1$: errors of SINH with $N_\ell=27$, $N^-=	43$, $N^+=	56$, $N^\pm_1=183$. CPU time 0.145 sec.\\
$\eps_2$: error of GWR with $M=8$, $N^-=51$, $N^+=67$, $N^\pm_1=309$. CPU time 0.143 sec.

}
\end{flushleft}

\label{table2nu0.2}
\end{table}

\begin{table}
\caption{\small Double barrier call. (Log-)barriers: $h_-=-0.05, h_+=0.05$; log-strike $a=0$.
KoBoL 
close to NIG, with an almost symmetric jump density, and no ``drift": $m_2=0.1, \nu=1.2, \lm=-2, \lp=1, \mu=0$. 
Riskless rate $r=0$.
Prices and errors (rounded) of the algorithm with the sinh- and GWR-acceleration applied to the Bromwich integral (SINH and
GWR).  Time to maturity $T=0.004,  0.25, 1, 3$.
 }
 {\tiny
\begin{tabular}{c|ccccc}
\hline\hline
$x$ & -0.04 & -0.02  & 0 & 0.02 & 0.04\\\hline\hline
\underline{$T=0.004$} & & & &  \\\hline
$V_{call}$ & 0.0000824404624213168 & 0.000204166217898571 & 0.00191177395953069 &
0.0197639609828199 & 0.0376099511198009\\
$\eps_1$ & -3.88E-10 &		-5.43E-10	&	-7.59E-10	& -1.25E-09	&	-3.68E-09\\
$\eps_2$ & 3.48E-11	&	-5.23E-10	&	4.67E-10	&	5.85E-10	&	-5.81E-09\\\hline

\underline{$T=0.25$} & & & &  \\\hline
$V_{call}$ & 0.00082409621567 & 0.00169255391432158 & 0.00212237950197215 & 0.00195192815759268
& 0.0010370403093618\\
$\eps_1$ & -2.24E-07	&	-3.00E-07	&-5.51E-07&	-4.50E-07	&	-4.68E-07\\
$\eps_2$ &  -1.92E-07	
&	-2.02E-07&	-4.01E-07	&	-2.26E-07	&	-1.29E-07\\\hline

\underline{$T=1$} & & & &  \\\hline
$V_{call}$ & 4.80103589597936E-06 & 9.53128787591073E-06 &0.0000112373996155185 & 9.72392418717438E-06
& 4.99704769890697E-06\\
$\eps_1$ & 5.93E-11	 &	7.07E-11	&	1.09E-10	&	2.40E-10	&	5.45E-10\\
$\eps_2$ & 4.41E-06 &		7.80E-06	&	7.86E-06	&	5.65E-06	&	2.50E-06\\
$\eps_3$ & 6.43E-13 &	9.80E-13	&	1.64E-12	&	3.09E-12	&	7.31E-12\\\hline
$T=3$ & & & &  \\\hline
$V_{call}$ & 4.074629523	&		8.089334758	&	9.537398649	&	8.252398764	&	4.240829909\\
$\eps$ & 1.67E-16&	-5.55E-17 &			-2.78E-17	& 2.78E-16	&	-4.44E-16\\
$rel_\eps$
& 1.39E-04	&	-8.33E-05	&	-2.78E-05	&	3.05E-04	&	-3.61E-04
\\\hline
 \end{tabular}
}
\begin{flushleft}{\tiny
\underline{\bf $T=0.004$.} Benchmark values at 9 points: CPU time 14.5 sec, precision is better than E-16. \\
Sizes of arrays for the benchmark:
$N_\ell=464$, $N^-=	180$, $N^+=	193$, $N^\pm_1=	1410$.\\
$\eps_1$: errors of SINH with $N_\ell=109$, $N^-=	73$, $N^+=	89$, $N^\pm_1=272$. CPU time 1.30 sec.\\
$\eps_2$: error of GWR with $M=8$, $N^-=58$, $N^+=78$, $N^\pm_1=349$. CPU time 0.128 sec.
\vskip0.1cm
\noindent 
\underline{\bf $T=0.25$.} Benchmark values at 9 points: CPU time 8.81 sec, precision is better than E-15. \\
Sizes of arrays for the benchmark:
$N_\ell=254$, $N^-=	175$, $N^+=	188$, $N^\pm_1=	1376$.\\
$\eps_1$: errors of SINH with $N_\ell=16$, $N^-=	27$, $N^+=	33$, $N^\pm_1=71$. CPU time 0.073 sec.\\
$\eps_2$: errors of GWR with $M=8$, $N^-=31$, $N^+=44$, $N^\pm_1=126$. CPU time 0.072 sec.

\vskip0.1cm
\noindent 
\underline{\bf $T=1$.} ``Benchmark values" (BB values) at 9 points: CPU time 9.82 sec, precision is better than E-15 if schemes with different deformations but the same $M_0=9$ are compared . \\
Sizes of arrays for the benchmark:
$N_\ell=276$, $N^-=	208$, $N^+=	222$, $N^\pm_1=	1312$.\\
$\eps_1$: errors of SINH with $N_\ell=59$, $N^-=	50$, $N^+=	66$, $N^\pm_1=968$. CPU time 0.317 sec.\\
$\eps_2$: errors of GWR with $M=8$, $N^-=152$, $N^+=162$, $N^\pm_1=126$. CPU time 0.521 sec.\\
$\eps_3$: errors of BB values w.r.t. to the scheme using IX(2) (matrix inversion instead of truncation of the series).
\vskip0.1cm
\noindent 
\underline{\bf $T=3$.}
Values ($V_{call}$) in units of $10^{-12}$, differences (in units of 1) and relative differences between values calculated using different deformations ($\eps$ and $rel_\eps$); step IX(2) is used.\\
Errors of GWR and Step IX(1) with any $M_0$ (and double precision arithmetic) are of the order of the values. 
}
\end{flushleft}
\label{table3nu1.2}
\end{table}

\begin{table}
\caption{\small Double barrier call. (Log-)barriers: $h_-=-0.05, h_+=0.05$; log-strike $a=0$.
KoBoL 
close to VG, with an almost symmetric jump density, and no ``drift": $m_2=0.1, \nu=0.2, \lm=-2, \lp=1, \mu=0$. 
Riskless rate $r=0$.
Prices and errors (rounded) of the algorithm with the sinh- and GWR-acceleration applied to the Bromwich integral (SINH and
GWR).  Time to maturity $T=0.004,  0.25, 3, 5$.
 }
 {\tiny
\begin{tabular}{c|ccccc}
\hline\hline
$x$ & -0.04 & -0.02  & 0 & 0.02 & 0.04\\\hline\hline
\underline{$T=0.004$} & & & &  \\\hline
$V_{call}$ & 8.77581627294074E-06 & 0.0000140126705130402 &
0.0000365669756568704 & 0.0201495491781489 & 0.0406572530052907\\
$\eps_1$ &1.88E-10	&3.45E-11	&	-3.96E-10	&	3.207E-12	&	-3.87E-10\\
$\eps_2$ & -1.45E-10	&	2.00E-10	&	4.10E-10	&	4.56E-10	&	1.02E-09\\\hline

\underline{$T=0.25$} & & & &  \\\hline
$V_{call}$ & 0.000461949526934386&
0.000737018983172683 &
0.001757476257127& 
0.0171789730510928 &
0.032229359938136
\\
$\eps_1$ & -3.62E-10	&	-5.38E-10	&	-8.39E-10	&	-6.50E-10	&	-1.15E-09
\\
$\eps_2$ & -6.14E-09 &		5.54E-09 &		9.84E-09	&	-5.37E-09	&	-5.705E-09
\\\hline

\underline{$T=3$ }& & & &  \\\hline
$V_{call}$ & 0.00080243923836716 & 0.00124070584073691 & 0.00182400986139031 & 0.00268682684255628 &
0.00274066739944123\\

$\eps_1$ & -6.20E-08	&	-6.93E-08	&	-7.87E-08	&	-9.20E-08	&	-1.17E-07\\
$\eps_2$ & 1.11E-08	 &	4.56E-08	&	-1.90E-07	&	-9.07E-08	&-1.388E-07\\\hline
$T=5$ & & & &  \\\hline
$V_{call}$ & 
0.000327733526265667 & 0.000493515765859709 & 0.000630026253687271
& 0.000720082418558615
& 0.000595670818077265\\
$\eps$ & 5.61E-12&	5.70E-12 &		5.802E-12	 &	5.90E-12	&	6.00E-12\\\hline
 \end{tabular}
}
\begin{flushleft}{\tiny
{\bf $T=0.004$.} Benchmark values at 9 points: CPU time 19.1 sec, precision is better than E-16. \\
Sizes of arrays for the benchmark:
$N_\ell=649$, $N^-=	173$, $N^+=	186$, $N^\pm_1=	1175$.\\

$\eps_1$: errors of SINH with $N_\ell=127$, $N^-=	54$, $N^+=	65$, $N^\pm_1=178$. CPU time 0.749 sec.\\
$\eps_2$: error of GWR with $M=8$, $N^-=52$, $N^+=68$, $N^\pm_1=207$. CPU time 0.113 sec.
\vskip0.1cm
\noindent 
{\bf $T=0.25$.} Benchmark values at 9 points: CPU time 7.63 sec, precision is better than E-16. \\
Sizes of arrays for the benchmark:
$N_\ell=255$, $N^-=	173$, $N^+=	186$, $N^\pm_1=	1175$.\\
$\eps_1$: errors of SINH with $N_\ell=79$, $N^-=	54$, $N^+=	65$, $N^\pm_1=178$. CPU time 0.449 sec.\\

$\eps_2$: error of GWR with $M=8$, $N^-=31$, $N^+=44$, $N^\pm_1=126$. CPU time 0.066 sec.

\vskip0.1cm
\noindent 
{\bf $T=3$.} Benchmark values at 9 points: CPU time 9.82 sec, precision is better than E-15.  Errors of BB values w.r.t. to the scheme using IX(2) (matrix inversion instead of truncation of the series) are marginally larger than E-15 at some points, and smaller than E-15 at other points.
Sizes of arrays for the benchmark:
$N_\ell=182$, $N^-=	173$, $N^+=	186$, $N^\pm_1=	1175$.\\
$\eps_1$: errors of SINH with $N_\ell=49$, $N^-=	59$, $N^+=	75$, $N^\pm_1=233$. CPU time 0.316 sec.\\
$\eps_2$: error of GWR with $M=8$, $N^-=38$, $N^+=51$, $N^\pm_1=143$. CPU time 0.096 sec.
\vskip0.1cm
\noindent 
\underline{\bf $T=5$.}
Values ($V_{call}$) and differences between values calculated using different deformations ($\eps$); Step IX(2) is used.
If GWR with Step IX(2) is used, the errors are of the order of E-07. \\
Errors of GWR and Step IX(1) with any $M_0$ (and double precision arithmetic) are of the order of the values. 

}
\end{flushleft}
\label{table3nu0.2}
\end{table}

\begin{figure}
\scalebox{0.75}
{\includegraphics{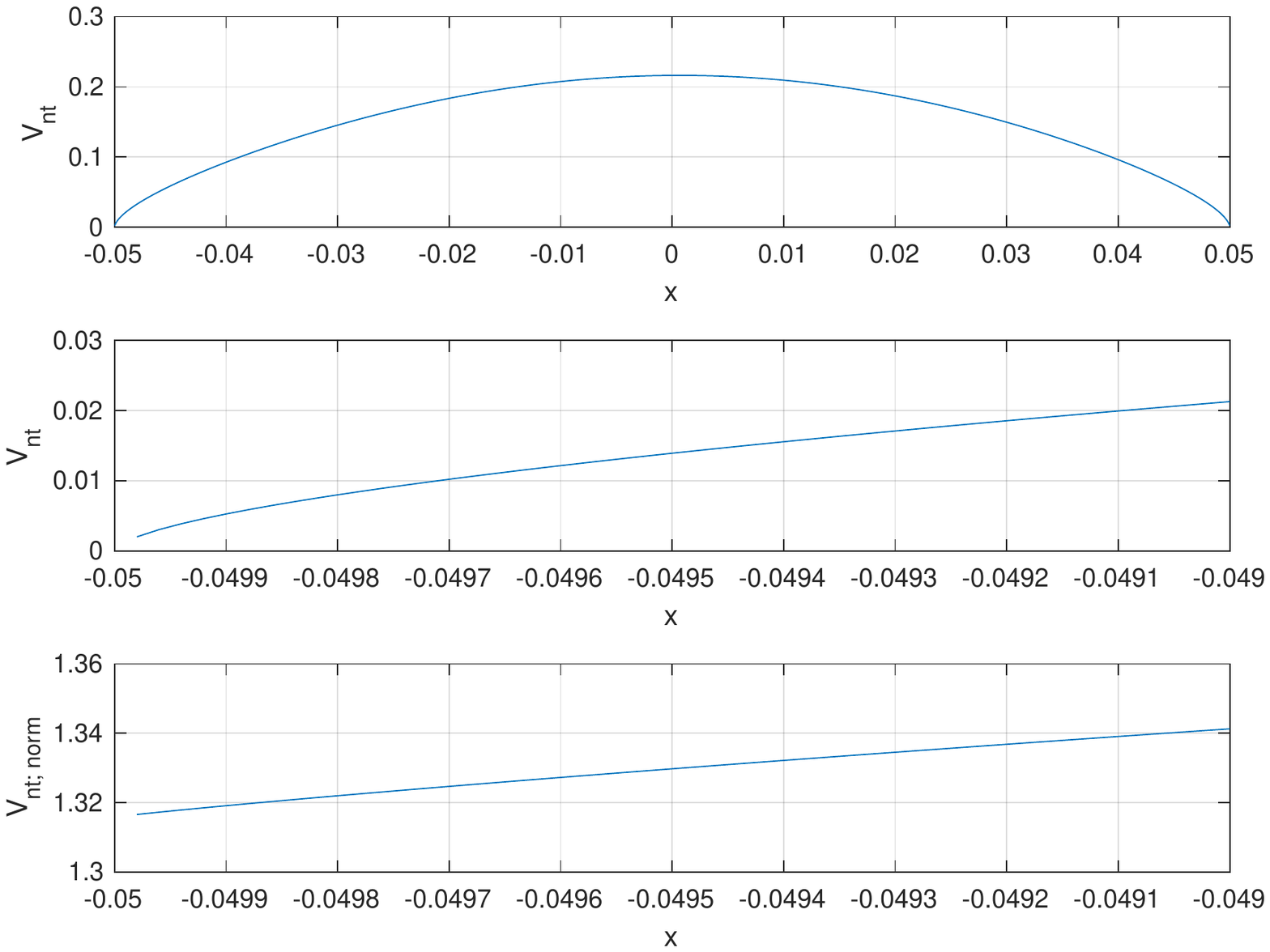}}
\caption{\small Double no-touch option, $h_-=-0.05, h_+=0.05$, $T=0.25$,  $\nu=1.2, m_2=0.1, \lp=1, \lm=-2, \mu=0$.}
\label{fig:DoubleNT4999_2}
\end{figure}

}
\end{example} 

\begin{figure}
\scalebox{0.75}
{\includegraphics{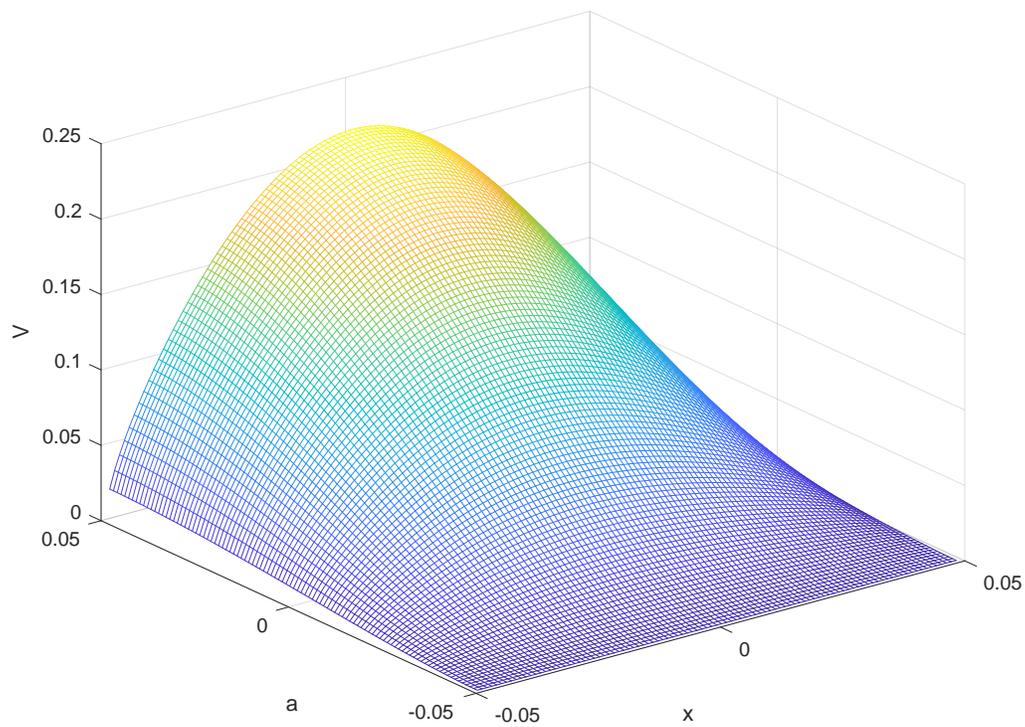}}
\caption{\small Double barrier digital, $h_-=-0.05, h_+=0.05$, $T=0.25$,  $\nu=1.2, m_2=0.1, \lp=1, \lm=-2, \mu=0$.}
\label{fig:VdoubleDig}
\end{figure}

\begin{figure}
\scalebox{0.75}
{\includegraphics{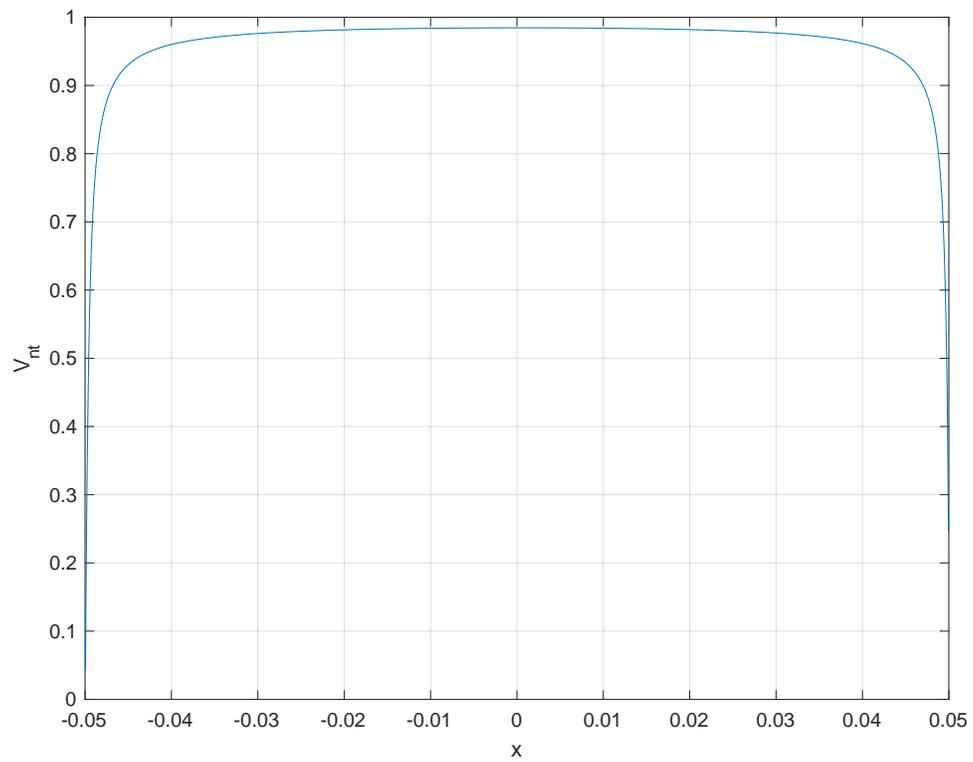}}
\caption{\small Double barrier no-touch option, $h_-=-0.05, h_+=0.05$,  $T=0.01$, $\nu=0.8, m_2=0.1, \lp=1, \lm=-2, \mu=0$.}
\label{fig:DoubleNTT001nu08mu002}
\end{figure}

\begin{figure}
\scalebox{0.75}
{\includegraphics{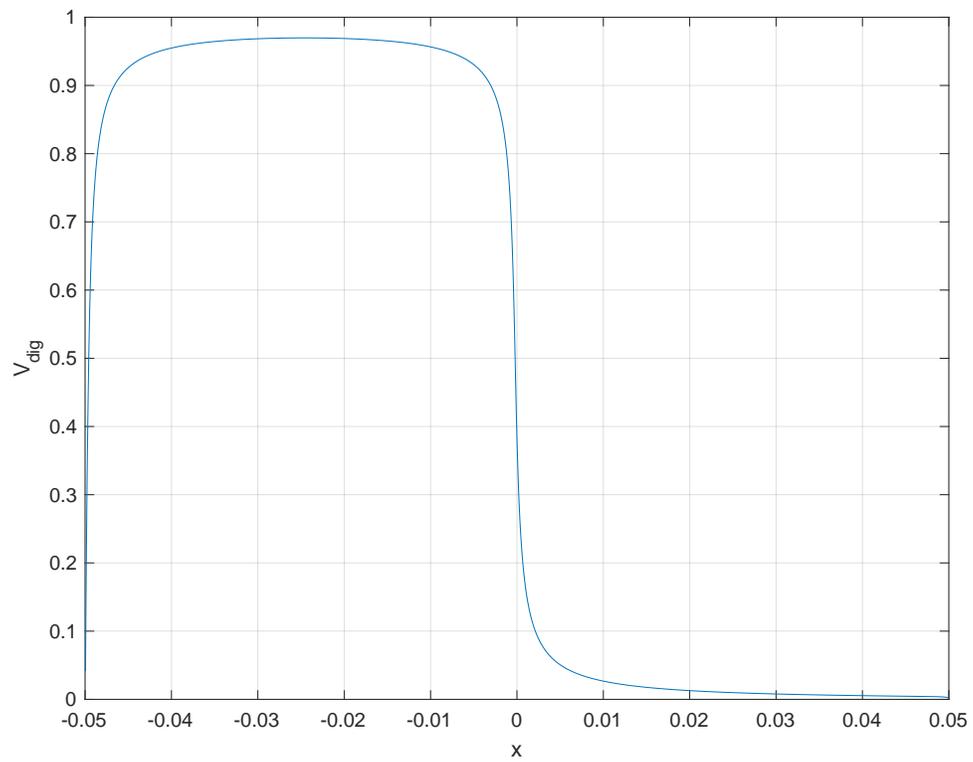}}
\caption{\small Double barrier digital, $h_-=-0.05, h_+=0.05$,  $T=0.01$, strike $K=e^0=1$, $\nu=0.8, m_2=0.1, \lp=1, \lm=-2, \mu=0$.}
\label{fig:DoubleDigT001nu008mu002a0}
\end{figure}

\begin{figure}
\scalebox{0.75}
{\includegraphics{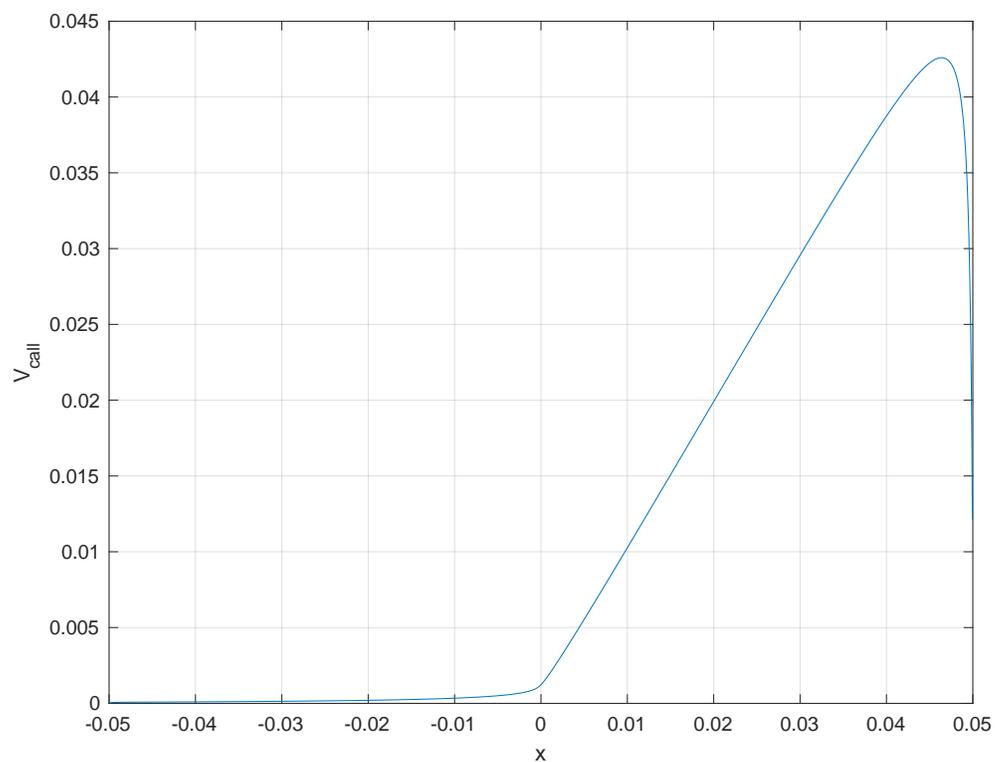}}
\caption{\small Double barrier call option, $h_-=-0.05, h_+=0.05$, $T=0.01$, strike $K=e^0=1$, $\nu=0.8, m_2=0.1, \lp=1, \lm=-2, \mu=0$.}
\label{fig:DoubleCallT001nu008mu002a0}
\end{figure}

\begin{figure}
\scalebox{0.75}
{\includegraphics{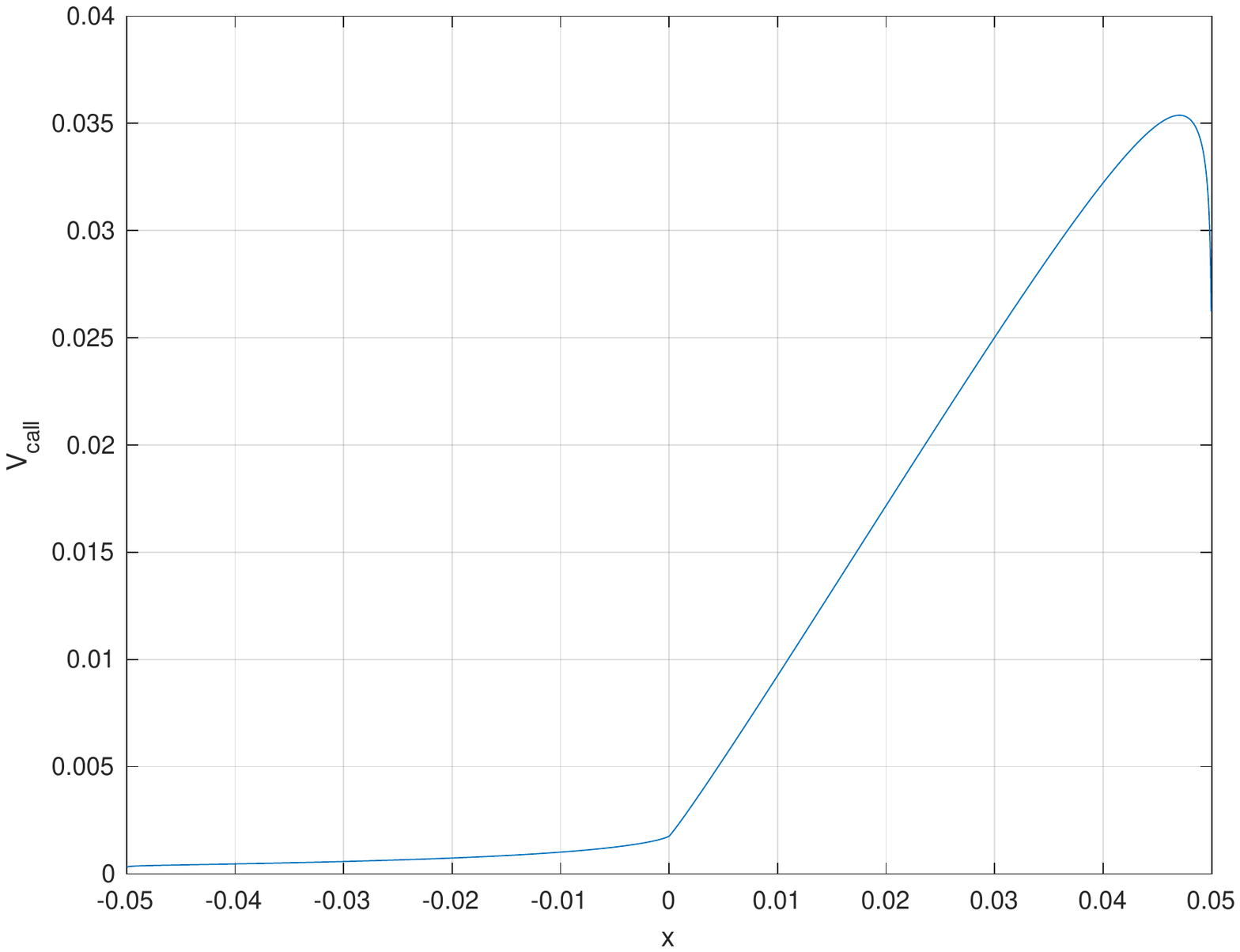}}
\caption{\small Double barrier call option, $h_-=-0.05, h_+=0.05$,  $T=0.25$, $K=1$, $\nu=0.2, m_2=0.1, \lp=1, \lm=-2, \mu=0$.} 
\label{fig:Vntnu02mu0T025}
\end{figure}

\begin{figure}
\scalebox{0.75}
{\includegraphics{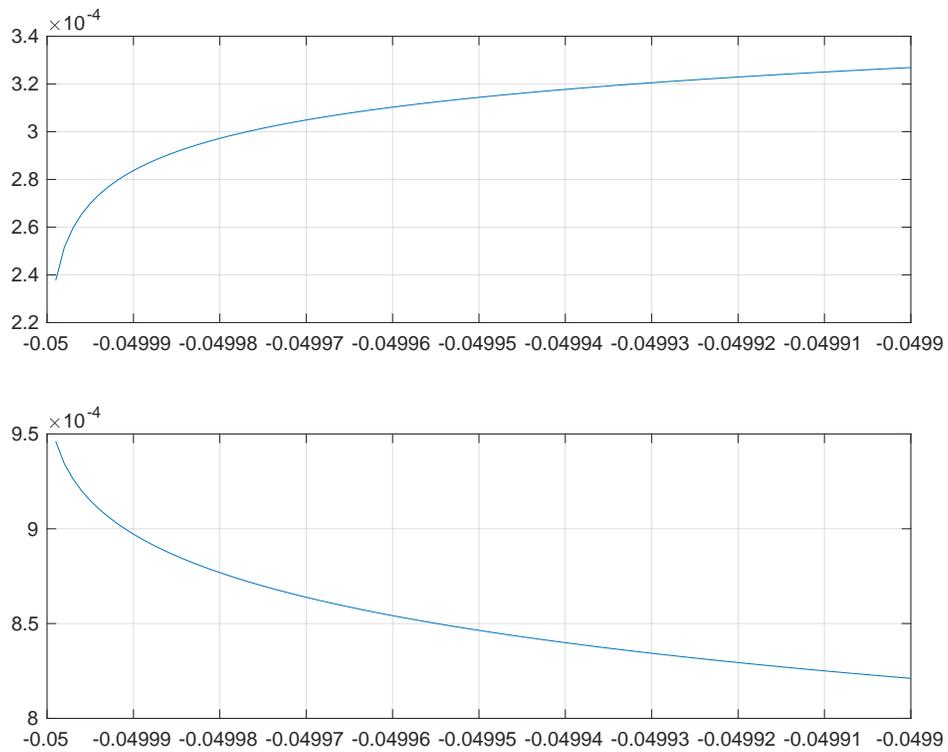}}
\caption{\small Double barrier call option, $h_-=-0.05, h_+=0.05$,  $T=0.25$, strike $K=e^0=1$. $K=1$, $\nu=0.2, m_2=0.1, \lp=1, \lm=-2, \mu=0$.
Upper panel; the call price. Lower panel: the normalized price $V_{call;norm}(x)=V_{call}(x)/(x-h_-)^{\nu/2}$}
\label{fig:Vntnu02mu0T025hm2}
\end{figure}
\begin{figure}
\scalebox{0.75}
{\includegraphics{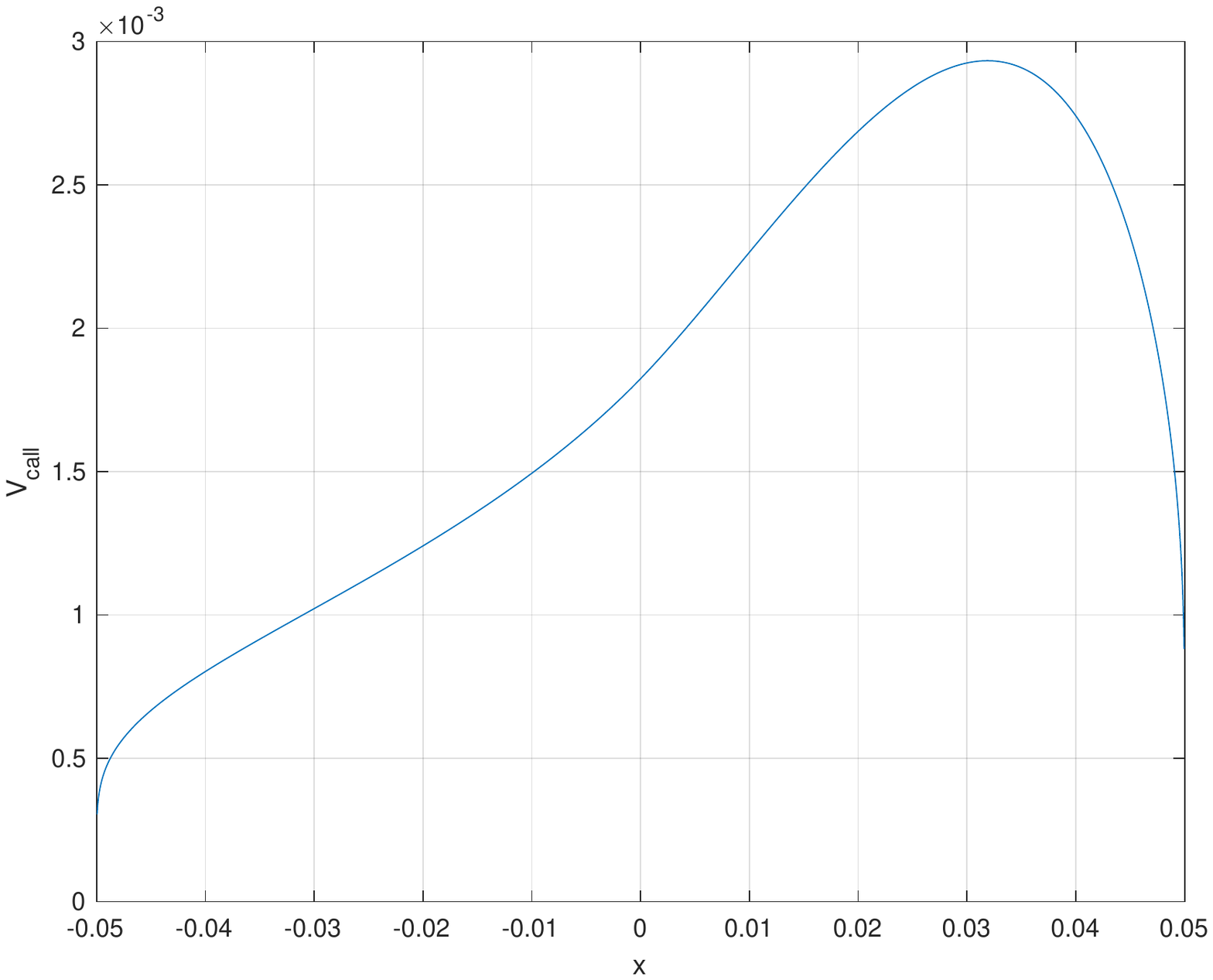}}
\caption{\small Double barrier call option, $h_-=-0.05, h_+=0.05$, $T=3$, $K=1$, $\nu=0.2, m_2=0.1, \lp=1, \lm=-2, \mu=0$.} 
\label{fig:Vntnu02mu0T3}
\end{figure}

\end{document}